\documentclass[final,5p,twocolumn,10pt]{elsarticle}
\pdfoutput=1 

\usepackage{lineno}
\usepackage{enumitem}
\modulolinenumbers[5]

\usepackage{amsmath}
\usepackage{amssymb}
\usepackage{stmaryrd}
\usepackage{amsthm}
\usepackage{amsfonts}
\usepackage{dsfont}
\usepackage{graphicx}
\usepackage{upgreek}
\usepackage{bm}
\usepackage{color}
\usepackage{float}
\usepackage{tikz-cd}
\usepackage[T1]{fontenc} 

\biboptions{numbers,sort&compress}

\newcommand{\bx}{{\mathbf{x}}}

\newcommand{\bJ}{{\mathbf{J}}}
\newcommand{\pose}{\tau}

\newcommand{\dil}{~\mathsf{dil}}
\newcommand{\ero}{\mathsf{ero}}
\newcommand{\sweep}{\mathsf{sweep}}
\newcommand{\unsweep}{\mathsf{unsweep}}
\newcommand{\obs}{\mathsf{obs}}

\newcommand{\powerset}{\mathcal{P}}

\newcommand{\cupr}{\mathrel{\cup}}
\newcommand{\capr}{\mathrel{\cap}}

\newcommand{\R}{{\mathds{R}}}

\newcommand{\SE}[1]{{\mathrm{SE}(#1)}}

\newcommand{\raw}{\mathsf{raw}}
\newcommand{\asM}{\mathsf{M}}
\newcommand{\asD}{\mathsf{D}}
\newcommand{\euc}{{\mathsf{E}}}
\newcommand{\conf}{{\mathsf{C}}}

\newcommand{\rset}{\mathbb{S}}
\newcommand{\rmotion}{\mathbb{C}}

\newcommand{\motions}{{\mathsf{T}}}
\newcommand{\mmfs}{{\mathsf{B}}}
\newcommand{\motion}{T}
\newcommand{\mmf}{B}
\newcommand{\rest}{C}
\newcommand{\inst}{D}
\newcommand{\action}{\varphi}
\newcommand{\process}{\Phi}

\newcommand{\expr}{\mathsf{E}}
\newcommand{\prim}{P}
\newcommand{\Prim}{\mathbb{P}}
\newcommand{\atom}{A}
\newcommand{\Atom}{\mathbb{A}}

\newcommand{\cons}{\mathrm{cons}}

\renewcommand{\th}{$^\text{th}$ }

\theoremstyle{definition}

\newtheorem{defn}{Definition}

\newtheorem{theo}{Theorem}

\newcommand{\eq}[1]{(\ref{#1})} 
\newcommand{\com}[1]{} 

\journal{\rm Computer-Aided Design (CAD), Special Issue on SPM'2018. DOI: \href{https://doi.org/10.1016/j.cad.2018.04.022}{10.1016/j.cad.2018.04.022}}

\begin{document}

\begin{frontmatter}

\title{Automated Process Planning for Hybrid Manufacturing}

\author{Morad Behandish, Saigopal Nelaturi, and Johan de Kleer}

\address{\rm
	Palo Alto Research Center (PARC),
	3333 Coyote Hill Road, Palo Alto, California 94304
	\vspace{-15.0pt}
}

\begin{abstract}

Hybrid manufacturing (HM) technologies combine additive and subtractive manufacturing (AM/SM) capabilities, leveraging AM's strengths in fabricating complex geometries and SM's precision and quality to produce finished parts. We present a systematic approach to automated computer-aided process planning (CAPP) for HM that can identify non-trivial, qualitatively distinct, and cost-optimal combinations of AM/SM modalities. A multimodal HM process plan is represented by a finite Boolean expression of AM and SM manufacturing primitives, such that the expression evaluates to an `as-manufactured' artifact.  We show that primitives that respect spatial constraints such as accessibility and collision avoidance may be constructed by solving inverse configuration space problems on the `as-designed' artifact and manufacturing instruments. The primitives generate a finite Boolean algebra (FBA) that enumerates the entire search space for planning. The FBA's canonical intersection terms (i.e., `atoms') provide the complete domain decomposition to reframe manufacturability analysis and process planning into purely symbolic reasoning, once a subcollection of atoms is found to be interchangeable with the design target. The approach subsumes unimodal (all-AM or all-SM) process planning as special cases. We demonstrate the practical potency of our framework and its computational efficiency when applied to process planning of complex 3D parts with dramatically different AM and SM instruments.

\end{abstract}


\end{frontmatter}

\date{May 18, 2018}


\section{Introduction} \label{sec_intro}

\noindent
Hybrid manufacturing (HM), combining the capabilities of additive and subtractive manufacturing, is the new frontier of part fabrication. While additive manufacturing (AM) continues to enable unprecedented levels of structural complexity and customization, subtractive manufacturing (SM) remains indispensable for producing high-precision, mission-critical, and reliable mechanical components with functional interfaces. Versatile `multi-tasking' machines with simultaneous high-axis computer numerical control (CNC) of multiple AM and SM instruments (e.g., deposition heads and cutting tools) keep emerging on the market, enabling efficient use-cases for fabrication and repair (reviewed in Section \ref{sec_lit}). It is only a matter of time before such processes dominate the shop floors as the unique and complementary benefits of AM and SM become vital to defense, aerospace, and consumer products.
 
Today, HM process planning rarely extends beyond the common ``AM-then-SM'' patterns. Use-case scenarios include support structure removal by CNC tooling after metal AM of near-net shapes and surface patching of corroded surfaces for repairing worn-out parts \cite{Liou2007applications,Ren2007part}. In some scenarios SM post-processing is inevitable due to the limitations of AM in producing overhang shapes or high-precision functional surfaces for assembly. In other scenarios it is a matter of saving production costs by optimizing material utilization---when AM/SM alone would require substantial material deposition/removal by starting from an empty platform or a large raw stock, respectively---or prolonging product lifecycles by using multiple alloys in a single part (e.g., Fig. \ref{fig_real} (g)). Such cases are already in use for producing corrosion-resistant parts for injection molding and oil transportation industries \cite{Yamazaki2016development}.

Although the capability to simultaneously use AM and SM exists in modern fabrication, there are very few examples of designs that are \emph{enabled} exclusively by HM. In most showcased success stories, the separation of features is trivial and the AM/SM actions come in predictable pairs that facilitate manual or semi-automatic process planning (e.g., Fig. \ref{fig_real} (a--f)). Even when designs enabled by HM can be conceptualized, planning their fabrication remains a manual activity driven by emerging expertise in HM.

This article presents theoretical foundations and computational algorithms to enable automatic construction of valid and cost-effective HM process plans for an arbitrary collection of AM/SM capabilities, provided by the same or different machine(s), with shapes and motions of arbitrary geometric complexity.

\begin{figure*}
	\centering \includegraphics[width=0.95\textwidth]{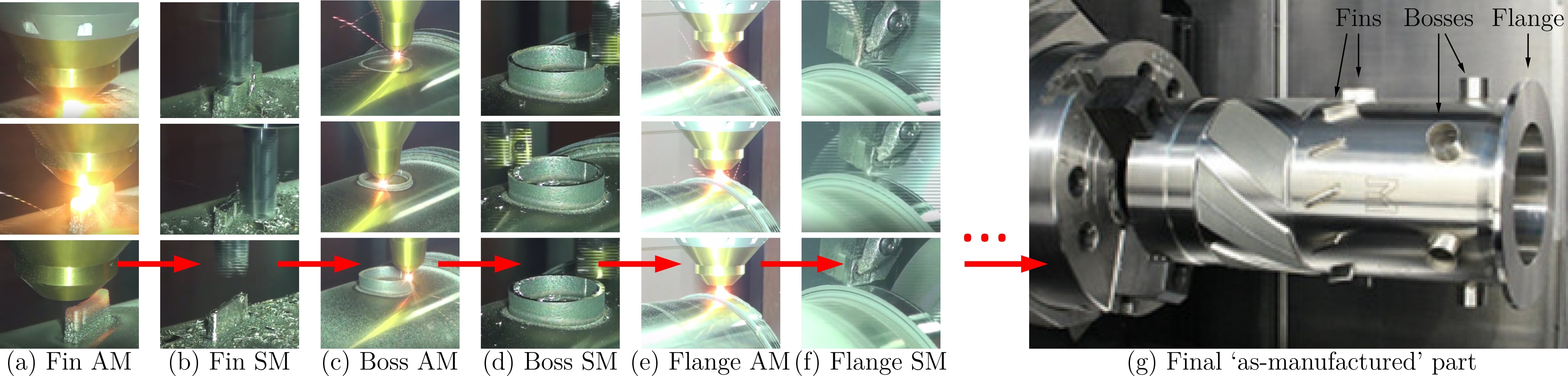}
	\caption{A metal part manufactured by a combination of 5-axis printing, milling, and turning operations on \textsf{Mazak INTEGREX i-400 AM} \cite{Yamazaki2016development}. The operations are typically planned in AM-then-SM pairs to grow features and finish them one-at-a-time. Source: \href{https://youtu.be/KbXJb4wcxnw}{youtu.be/KbXJb4wcxnw}.} \label{fig_real}
\end{figure*}

\subsection{Related Work} \label{sec_lit}

Recently a number of  manufacturing studies have reported on ``hybridizing'' select AM and SM capabilities \cite{Jeng2001mold,Liou2001research}. Among the successful concepts are hybrid layered manufacturing (HLM) \cite{Akula2006hybrid,Karunakaran2010low} and surface patching \cite{Liou2007applications,Ren2007part} that combine selective laser cladding (SLC) and CNC machining for rapid prototyping (RP), repair and modification of die/mold parts, and re-tipping of high-value aerospace turbine blades \cite{Jones2012remanufacture}. Other combinations include SLC and CNC mill-turning \cite{Yamazaki2016development} as well as direct metal laser sintering (DMLS) and precision milling \cite{Du2016novel}. For reviews of HM technologies available today, see \cite{Zhu2013review,Lorenz2015review,Flynn2016hybrid,Merklein2016hybrid}.

As HM hardware technologies are striding ahead, computer aided process planning (CAPP) software tools to support their incredible potential are falling behind. Among the few reported efforts, Manogharan et al. \cite{Manogharan2015aims} introduced a HM system whose software component collected a suite of existing tools used in pure AM/SM process planning such as visibility analysis, fixture design, deviation/over-growth quantification, and tool-path planning, without addressing spatial complications that are unique to commingled AM+SM. \textsf{Siemens PLM Software} is now offering HM computer-aided manufacturing (CAM) tools as part of its \textsf{NX} solutions \cite{Siemens2014hybrid}, also using feature-based decomposition into pure AM/SM segments, each to be independently path planned. To the best of our knowledge, none of the existing software tools are able to systematically explore alternative HM process plans where the same 3D regions of a part---not necessarily separable as a standalone feature ---can be made with both AM/SM, and make cost-driven decisions.

Automatic feature recognition comprises a large body of literature for traditional SM (reviewed in \cite{Shah1991survey,Subrahmanyam1995overview,Gupta1997automated,Han2000manufacturing}). Notable techniques include volumetric decomposition \cite{Woo1982feature,Kim1992recognition}, graph-based B-rep analysis \cite{Joshi1988graph,Wu1996analysis}, and rule-based pattern recognition \cite{Vandenbrande1993spatial,Babic2008review} among others. Despite being effective when features are clearly separable, these methods do not extend to complex shapes with unclassifiable or interacting/intersecting features \cite{Tseng1994recognizing}. The notion of a ``feature''---one that depends on engineering intent \cite{Shah1991survey} with no consistent definition across design and manufacturing---is even more ambiguous in AM, leading to knowledge-based ontologies with their own limitations of applicability.

Recently, our group presented an alternative, {\it feature-free} method for CNC milling based on maximal machinable volumes in any accessible orientation \cite{Nelaturi2015automatic}, enabling a rapid process planning paradigm that scales to part/tool shapes and motions of arbitrary complexity (Fig. \ref{fig_smcost}). The underlying mathematical foundations were later shown to be applicable to AM analysis and design correction/feedback \cite{Nelaturi2015manufacturability} as well. There are a number of fundamental challenges in extending these ideas to HM processes with interleaved AM/SM actions that we discuss in Section \ref{sec_approach}.

\begin{figure}
	\centering
	\includegraphics[width=0.45\textwidth]{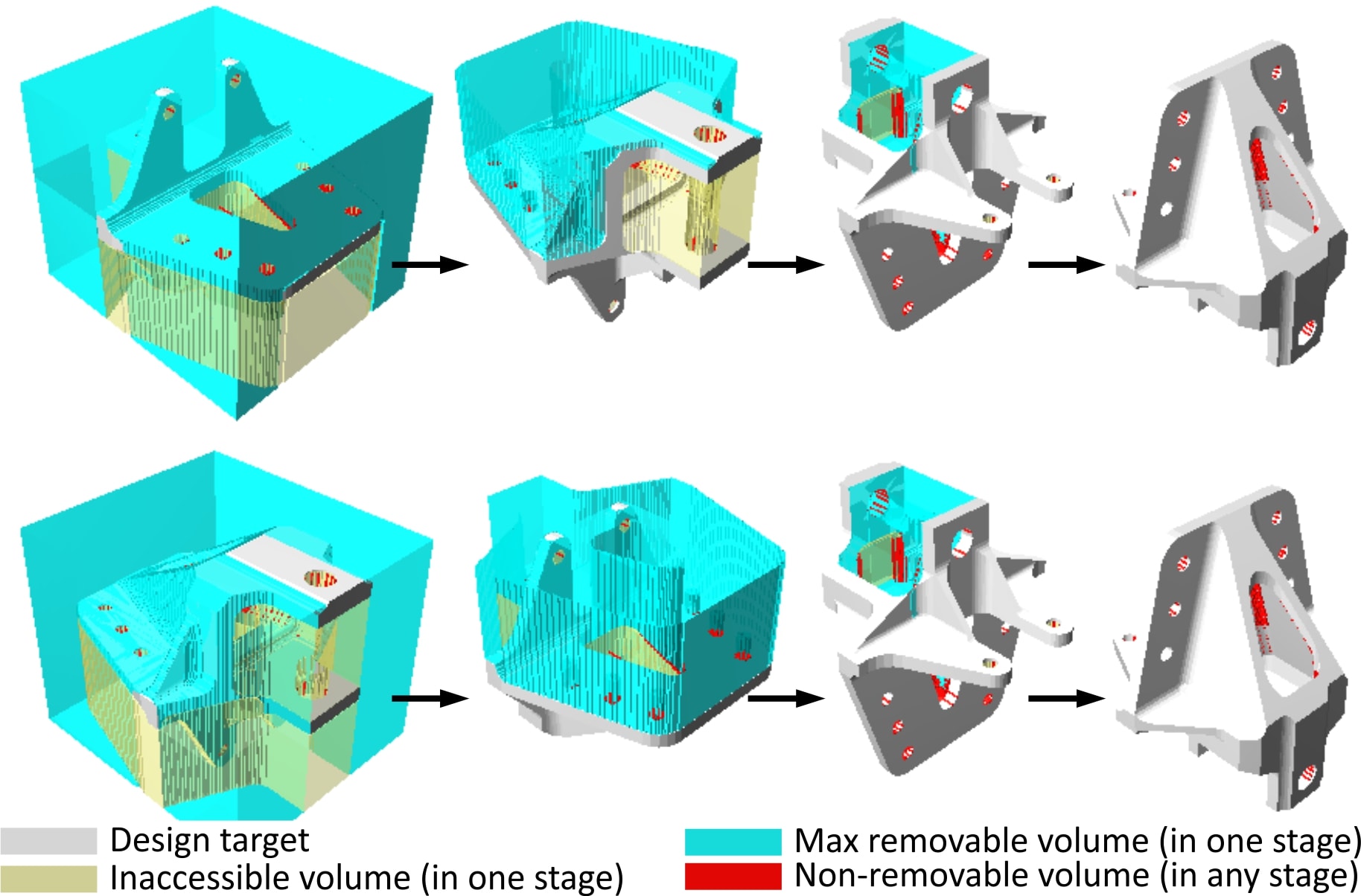}
	\caption{Qualitatively distinct SM plans with different costs. Notice the sequence of orientations in which volumes are removed differs between the plans, despite converging to the same final shape.} \label{fig_smcost}
\end{figure}

The current approach subsumes our earlier work in machining process planning \cite{Nelaturi2015automatic} as a special case. A major breakthrough was brought about by the ability to enumerate the entire search space using a logical (rather than geometric) representation in terms of a finite Boolean algebra (FBA), which enables {\it formulating and solving the planning problem in purely symbolic terms.} Rather than storing geometric representations of the evolving workpiece and the removal volumes associated with each manufacturing action (Fig. \ref{fig_smcost}), we show that such information can be encoded as binary strings in terms of the atomic units of manufacturing.

\subsection{Contributions \& Outline}

This article presents a computational framework to evaluate manufacturability and find process plans for HM. It accommodates a large class of existing (and potentially future) AM/SM capabilities by an abstraction that separates geometric and spatial reasoning for accessibility analysis and collision avoidance from logical and symbolic reasoning used to search for optimal plans. We show that:
\begin{enumerate}
	\item HM processes can be geometrically described by sequences of idempotent AM/SM actions (Section \ref{sec_proc}).

	\item The AM/SM primitives characterizing the actions can be constructed independently as the {\it closest} shapes to the design target achievable by means of a single AM/SM capability in a particular setup (Section \ref{sec_cap}).

	\item A set-theoretic notion of ``closeness'' is formulated with respect to minimal/maximal deposition/removal volumes, and computed using group morphological operations \cite{Lysenko2010group} in the configuration space of relative part/tool motions (Section \ref{sec_prim}).

	\item Manufacturability tests can be performed prior to the costlier process planning, using canonical representations of the design target in the FBA of the aforementioned primitives (Section \ref{sec_decomp}).

	\item The search space for process planning is describable in purely symbolic terms with respect to FBA states and transitions, and is explored using standard AI search algorithms \cite{Korf2010algorithms} (Section \ref{sec_hmplans}).
\end{enumerate}

\section{An  Overview Of the Approach} \label{sec_approach}

\begin{figure*}
	\centering \includegraphics[width=0.95\textwidth]{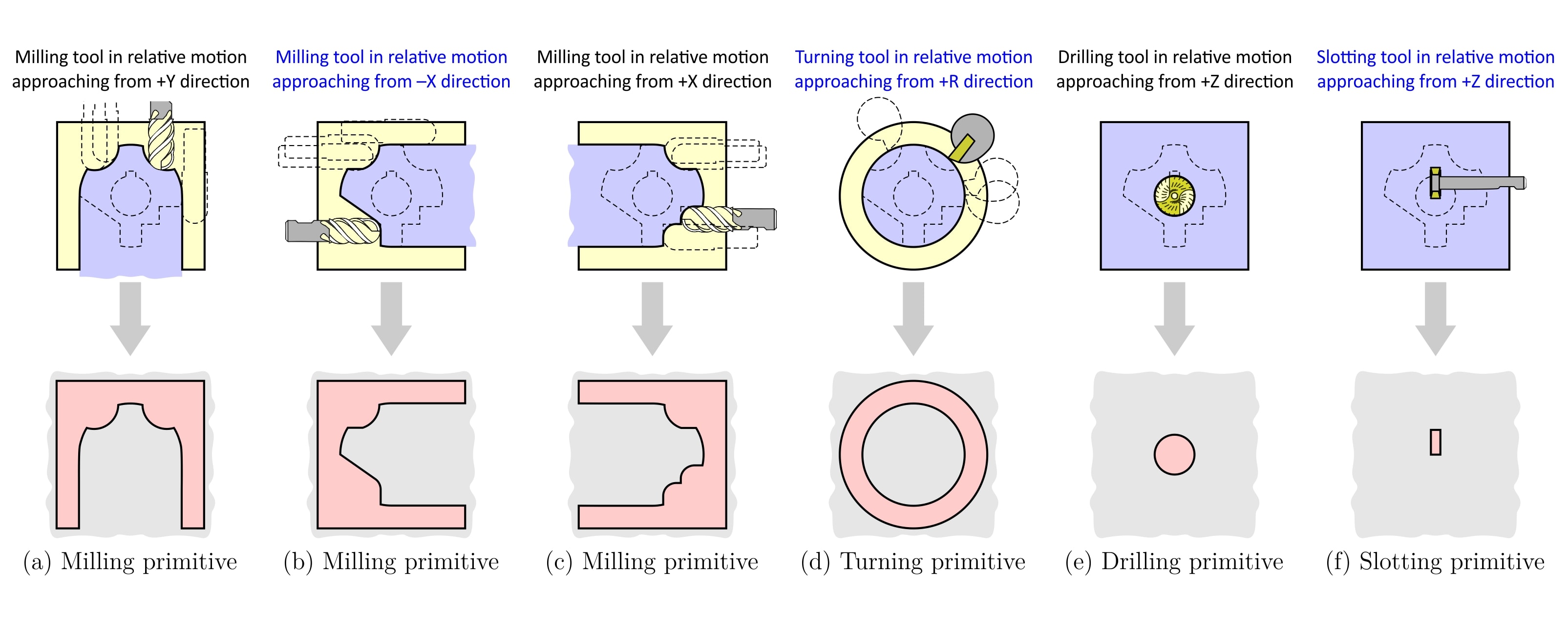}
	\centering \includegraphics[width=0.95\textwidth]{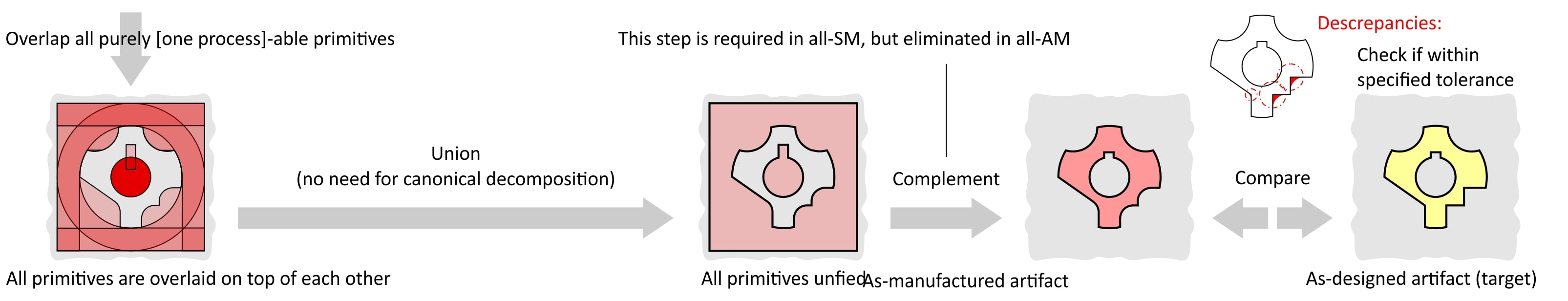}
	\caption{Our workflow for manufacturability analysis for unimodal processes (all-AM or all-SM, the latter in this example). Compare the workflow with that of multimodal manufacturability analysis in Fig. \ref{fig_mult} (bottom) .} \label{fig_uni}
\end{figure*}

\begin{figure*}
	\centering \includegraphics[width=0.95\textwidth]{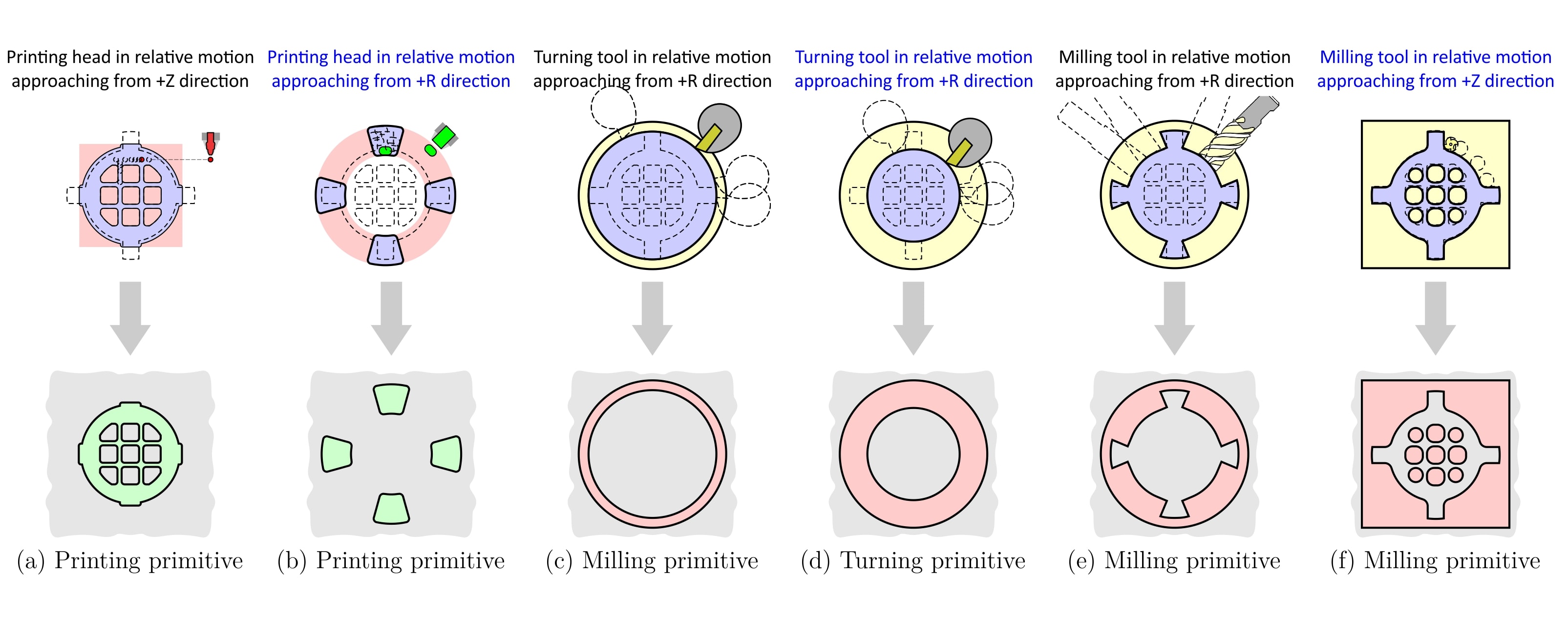}
	\centering \includegraphics[width=0.95\textwidth]{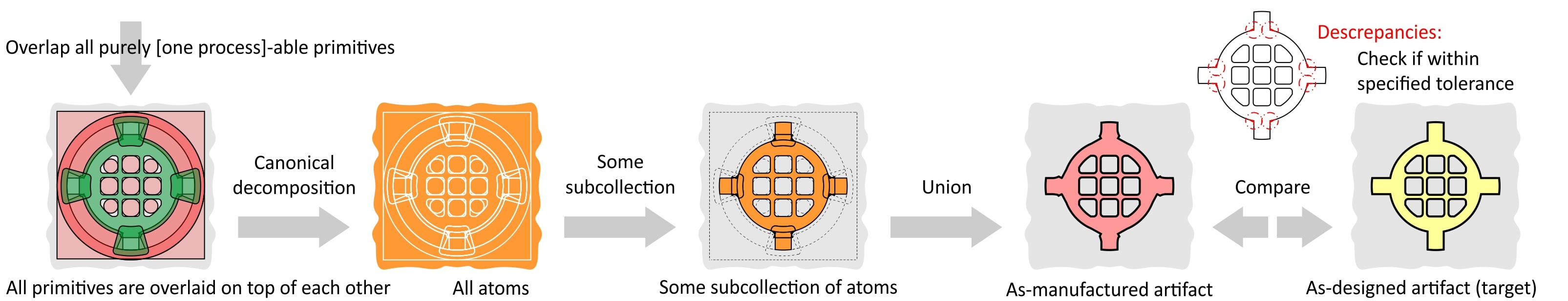}
	\caption{Our workflow for manufacturability analysis for multimodal processes (AM and SM used in arbitrary alternations). Compare the workflow with that of unimodal manufacturability analysis in Fig. \ref{fig_uni} (bottom) .} \label{fig_mult}
\end{figure*}

\subsection{Manufacturing Actions and Regions of Influence}

We view a manufacturing {\it process} as a finite sequence of manufacturing (either AM or SM) {\it actions} each defined by deposition/removal of a piece of material onto/from a part (Definitions \ref{def_action} and \ref{def_process}). The actions are idempotent---i.e., repetitive execution of any action is redundant. Each AM/SM action has a pre-defined `region of influence' (ROI) within which its material deposition/removal effect on the part's intermediate state (hereafter called the `workpiece') does not depend on whether or not there was material in the ROI beforehand, and outside which the action has no effect. However, the amount of material deposition/removal (hence the {\it cost}) does depend on preconditions inside the ROI, and the process planner has to take that into account.

The actions are {\it quanta of manufacturing}, whose ROIs are characterized by manufacturable portions of a part's interior/exterior for AM/SM, respectively. ROI shapes may be constrained by any number of conditions such as accessibility, collision avoidance, manufacturing resolution, operation-specific rules such as overhang angle thresholds, and so on. In this paper, ROI shapes are affected by two fundamental notions; namely, the machine's degrees of freedom (DOF) and the manufacturing instrument's minimum manufacturable neighborhoods (MMN) (Definition \ref{def_capability}). The former determines the class of relative {\it motions} that a CNC motion system supports, while the latter determines the {\it shapes} of the smallest features of ROI that may result from an action, which, in turn, depends on laser beam diameter, nozzle cross-section, cutter insert profile, etc. We show how each action's ROI is computed as the 3D pointset that can be swept by the MMN's shape, along a motion that respects the DOF, without colliding the instrument with surrounding obstacles or violating operation-specific rules.

One can think of ROIs as manufacturing {\it primitives} (Definition \ref{def_primitive}), using common terminology in constructive solid geometry (CSG) \cite{Requicha1977constructive}. Planning and search algorithms are then invoked to sequence the actions in a cost-optimal fashion by considering the effect of incrementally adding/removing material confined to each action's ROI. A plan may thus be represented as a Boolean formula in terms of ROIs (i.e., {\it as-planned} view), and is valid if its evaluation (i.e., {\it as-manufactured} view) results in a solid model that is equivalent to the design target (i.e., \emph{as-designed} view). The equivalence, commonly refereed to as part `interchangeability', is typically certified with respect to geometric dimensioning and tolerancing (GD\&T) \cite{ASME2009dimensioning}.

\subsection{Challenges Unique to Hybrid Manufacturing}

To appreciate the challenging nature of multimodal (i.e., interleaved AM-and-SM) processes, it is useful to first understand the important properties that substantially simplify unimodal (i.e., all-AM or all-SM) processes:
\begin{description}[style=unboxed,leftmargin=0cm] \setlength{\parskip}{2pt}
	\item[Monotonicity:] Every added action takes the state of the workpiece one step closer (in size) to the target in terms of deposited/removed material. Therefore the intermediate states of the part consistently increase (all-AM) or decrease (all-SM) in size, with no chance of going back. This guarantees that plans will terminate without any chance of undesirable (possibly never-ending) cycles of adding and removing the same material. Knowing that a process is monotonic justifies defining ROIs as the ``maximal'' depositable/removable regions (MDR/MRR), maximality being implied in terms of set containment, subject to manufacturing constraints---e.g., using a given nozzle/tool, at a fixed build/fixture orientation. We have shown in \cite{Nelaturi2015automatic} that MRR computations for milling of arbitrarily complex part/tool geometries can be posed as {\it inverse} configuration space problems and used to solve unimodal machining process planning problems (Fig. \ref{fig_smcost}). Similar methods were shown to compute MDRs for all-AM \cite{Nelaturi2015manufacturability}.
	
	\item[Permutativity:] Given a fixed set of actions/primitives, their cumulative effect (in terms of the {\it shape}) is invariant under permutation, i.e., insensitive to the {\it order} of execution. Permutativity of unimodal processes enables convenient decoupling of the two main tasks of manufacturability test and process planning. Both unions (AM) and intersections (SM) are commutative and associative among themselves (but not with each other). As a result, the unordered collection of primitives provides sufficient information to rapidly determine whether their application in an arbitrary order produces a part whose form is interchangeable with the target design, prior to the costlier task of planning. This is incredibly valuable as it can quickly detect if a collection of AM/SM primitives are (in)sufficient to make a part in a unimodal sequence, without spending computational resources to explore a dead-end.
\end{description}

Hybrid (i.e., multimodal) processes are neither monotonic nor permutative, making their analysis and planning more difficult. One solution is to break a multimodal sequence into unimodal sub-sequences, i.e., consecutive trains of unions or intersections, over which the said properties hold. The challenge is that interleaving AM/SM modalities as such does not suggest a path to an early manufacturability test, and the naive approach of nesting unimodal planners leads to prohibitive computing times.

\subsection{Early Manufacturability Test for HM Processes}

A more powerful framework than nesting emerges from a set-theoretic analysis of Boolean expressions, the likes of which have been used in  fundamental solid modeling applications ranging from representation conversion and maintenance \cite{Shapiro1991efficient,Shapiro1993separation} to design under spatial constraints \cite{Ilies2000shaping}. Every finite collection of primitives decomposes the 3D space into a finite number of disjoint {\it atomic} cells, obtained by intersecting all combinations of primitives and their complements \cite{Shapiro1997maintenance}.
Using arguments from basic set theory and logic, {\it the outcome of every finite Boolean expression is the union of some subcollection of atoms}. Physically valid manufacturing plans are a special class of such expressions, thus not every subcollection of atoms is manufacturable, but a shape is not manufacturable if it is not composable from whole atoms. This offers a path forward for an early (non-)manufacturability test for multimodal processes, notwithstanding the loss of permutativity. Given a target design, an early test checks if every atom is either completely inside or completely outside, up to discrepancies smaller than tolerance specs (Theorem \ref{theo_dnf}). The test is decisive for {\it non}-manufacturability, i.e., if it fails, the set of predefined actions cannot possibly make the part in any order of execution. As a necessary (though insufficient) condition, it helps avoiding inevitable dead-ends prior to spending computational resources on planning. More importantly, if new actions are to be added, the test provides insight as to what the new primitives should look like. If the test passes, manufacturability is not guaranteed, and ``false positives'' need to be ruled out during planning.

Figures \ref{fig_uni} and \ref{fig_mult} show the difference in the workflows for unimodal and multimodal manufacturability analysis.

To keep the discussion focused, we restrict ourselves to a high-level view of CAPP that has mainly to do with geometric and spatial complications of commingling AM/SM actions. This warrants a solid modeling approach to a first-order abstraction of manufacturing actions. If one or more qualitatively distinct sequences of actions are found to be (at least) geometrically  viable, lower-level analyses of tool-path planning and machine control as well as physical (e.g., mechanical/thermal) process simulation can be addressed subsequently. These problems are  typically confined to a single action's ROI and can be used to (in)validate candidate plans or  shortlist them based on domain-specific knowledge of physical/material constraints.

\section{Hybrid Manufacturing Processes} \label{sec_proc}

We choose the state space of the evolving workpiece as the collection of all 3D solids (i.e., compact-regular and semianalytic 3D sets, or `r-sets') \cite{Requicha1977mathematical}, denoted $\rset \subset \powerset(\euc)$ where $\euc = \R^3$ is the Euclidean $3-$space.%
	\footnote{$\powerset(U) = \{ X~|~ X \subseteq U \}$ stands for the powerset of a universe $U$.}
Motions used to sweep the MDR/MRR and construct ROIs are restricted to the collection of ``well-behaved'' \cite{Behandish2017analytic} subsets  $\rmotion \subset \powerset(\conf)$ of the Lie group $\conf := \SE{3}$ of rigid transformations, also called the {\it configuration space} ($\conf-$space) \cite{Lozano-Perez1983spatial}. The action of a rigid transformation (i.e., configuration) $\tau \in \conf$ on a single point $\bx \in \euc$ is denoted by $\pose \bx \in \euc$. By extension, applying a configuration to a pointset (including a solid) $S \subseteq \euc$ is denoted by $\pose S := \{ \pose \bx ~|~ \bx \in S \}$.

\begin{defn} (Manufacturing Action) \label{def_action}
	A manufacturing `action' is a mapping $\action: \rset \to \rset$ such that for
	all $S \in \rset$:
	\begin{align}
		\action^\uparrow(S) = (S \cupr \prim), \quad\text{for AM actions (i.e., $\prim \in \Prim^\uparrow$)}, \label{eq_action_am} \\
		\action^\downarrow(S) = (S \capr \overline{\prim}), \quad\text{for SM actions (i.e., $\prim \in \Prim^\downarrow$)}. \label{eq_action_sm}
	\end{align}
	for some fixed $\prim \in \rset$ called a manufacturing `primitive', marked AM or SM accordingly. $\cupr, \capr : \rset \times \rset \to \rset$ and $\overline{(\cdot)}: \rset \to \rset$ denote regularized%
		\footnote{To keep the notation simple, we use $\cupr, \capr, -$ for regularized operations commonly denoted $\cup^\ast, \cap^\ast, -^\ast$ in the literature \cite{Requicha1977constructive,Requicha1977mathematical}.}
	union, intersection, and complement operators to ensure algebraic closure of manufacturable states under finite sequences of actions.
\end{defn} 

\begin{defn} (Manufacturing Process) \label{def_process}
	A manufacturing `process' is a mapping $\process: \rset \to \rset$ that can be
	described by at least one finite composition of actions:
	\begin{equation}
		\process = (\action_n \circ \action_{n-1} \circ \cdots \circ \action_1), \quad
		\text{for some}~ n \geq 1. \label{eq_act_compose}
	\end{equation}
	A process is called `unimodal' if its actions are marked either all-AM or
	all-SM, and `multimodal' otherwise.
\end{defn}

A process can be expressed as state transitions from an initial state $S_0$ to a
final desired state $S_n = S_\asM$:
\begin{equation}
	S_0 \overset{\action_1}{\longrightarrow} S_1
	\overset{\action_2}{\longrightarrow} S_2 \overset{\action_3}{\longrightarrow}
	\cdots \overset{\action_n}{\longrightarrow} S_n = S_\asM . \label{eq_seq}
\end{equation}
A part is called `manufacturable' by a given set of actions if a sequence like
this exists such that the as-manufactured part $S_\asM$ is interchangeable with
the as-designed part $S_\asD$ (e.g., certified by GD\&T specifications).

Note that $\action^\uparrow(S) \supseteq S$ for AM whereas
$\action^\downarrow(S) \subseteq S$ for SM actions. Without loss of generality,
we assume $S_0 = \emptyset$ and the first action $\action_1$ is always additive.
This includes selecting a raw/bar stock $S_1 = \action^\uparrow_1(\emptyset) =
P_1$ as well as material deposition (e.g., 3D printing) on an empty platform.

\section{Hybrid Manufacturing Capabilities} \label{sec_cap}

\begin{figure*}
	\centering
	\includegraphics[width=0.95\textwidth]{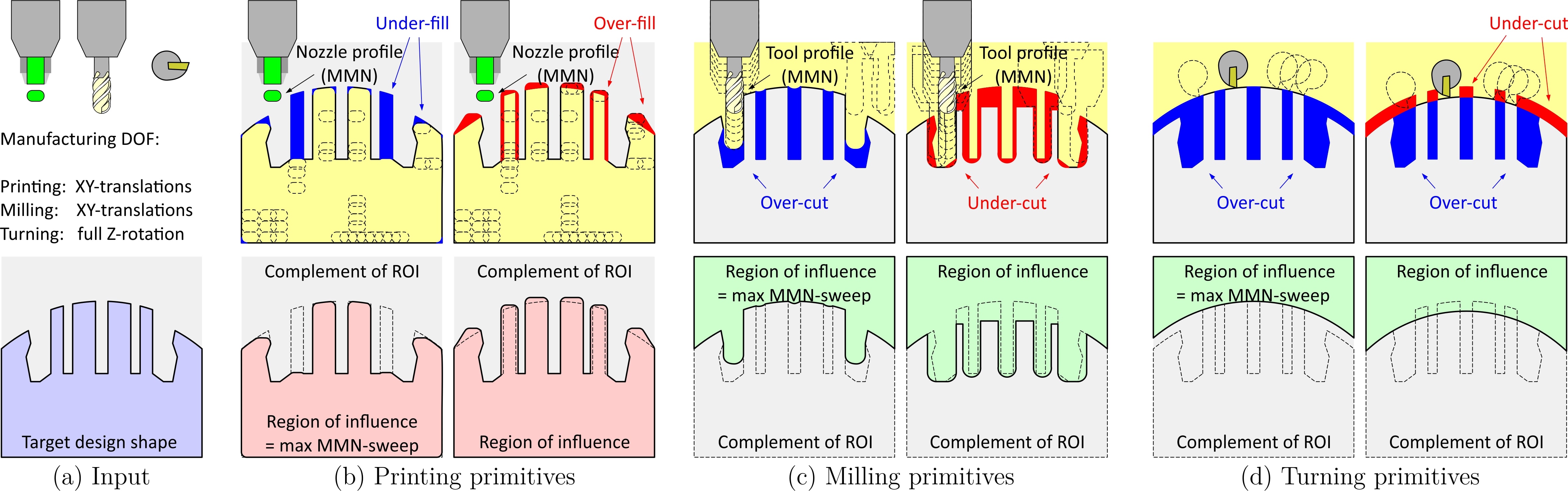} \caption{Manufacturing primitives obtained for 3 hybrid manufacturing capabilities (a) with U-AM/O-AM (b) O-SM/U-SM (c) and combination of the latter (d). Each primitive is obtained by sweeping the MMN along an allowable motion to get as close to the target as possible with a single capability. Different metrics can be used to quantify the ``fit'' between primitives and the target.} \label{fig_primitives}
\end{figure*} 

Informally, a manufacturing primitive is a 3D solid that is created by {\it exclusively one} manufacturing capability (e.g., 3D printing, turning, or milling) using a given instrument (e.g., nozzle or tool insert geometry) in a given fixed setup (e.g., build or fixture configuration), and is {\it closest} to the as-designed target $S_\asD$ with respect to a precisely defined metric for manufacturing progress. In other words, a primitive characterizes the best shot one can take using a single (either AM or SM) capability, to transition from a given state to a closer state to the target state, while {\it respecting the manufacturing constraints of that manufacturing capability.} The closeness to the target is often characterized by a volumetric difference. For example, the cost of a machining action is often approximated as the volume removed by that action divided by the tool's material removal rate (Fig. \ref{fig_smcost}) \cite{Nelaturi2015automatic}. Similarly, the cost of a deposition action can be approximated by the volume deposited in a tool path divided by the nozzle feed-rate \cite{Nelaturi2015manufacturability}.

Manufacturing constraints are abstracted by 1) degrees of freedom (DOF); and 2) minimum manufacturable neighborhoods (MMN) associated with a manufacturing instrument. The manufacturing DOF is a characterization of how different instruments can move; for instance:
\begin{itemize}
	\item Most 3D printing techniques are characterized by 3D translations, restricted to 2D horizontal layers.

	\item $3-$axis milling is characterized by 3D translations, often restricted to 2D surfaces depending on the cutting mode (e.g., end-, profile-, or contour-milling).
	
	\item Uniaxial turning is characterized by full (i.e., $360^\circ$) rotations combined with independent longitudinal and radial translations (e.g., side-, face-, or contour-turning) or screw motions, which are proportional rotations and translations (e.g., thread-turning).
\end{itemize}

As the instrument moves within these restricted regimes, its active part contacts the workpiece and modifies the workpiece's state by either of material deposition/removal, both of which can be geometrically simulated by {\it sweeping} the MMN along the motion. For example:
\begin{itemize}
	\item 3D printing is emulated by sweeping a shape that represents a droplet of solidified material, roughly the same size as the extrusion of a nozzle or laser beam cross-section by layer thickness.

\item $3-$axis milling is emulated by sweeping the rotating closure of a milling tool along the 3D motions.

\item Uniaxial turning is emulated by a rotational sweep (i.e., axisymmetric revolution) or helicoidal sweep of a tool insert around and along the turning axis.
\end{itemize}
The abstraction of a manufacturing capability based on DOF and MMN is sufficiently general to encompass many existing AM (e.g., SLA, FDM, SLS, SLM, and EBM), existing SM (e.g., turning, milling, drilling, and EDM), and conceivably, future capabilities that operate based on swept regions by relative spatial motions. Next, we make these notions precise for high-DOF motions.

\begin{defn} (Manufacturing Capability) \label{def_capability}
	A manufacturing `capability' is a product space $(\motions \times \mmfs)$ where
	\begin{itemize}
		\item $\motions \subseteq \rmotion$ is a collection of all allowable motions of a manufacturing instrument determined by the machine DOF and workspace size.

		\item $\mmfs \subseteq \rset$ is a collection of all available shapes of the MMN, which are inferred from the shapes of nozzles, laser beams, tool profiles, and so on.
	\end{itemize}
\end{defn}
A manufacturing capability is instantiated as a pair $(\motion, \mmf) \in (\motions \times \mmfs)$. For example, a turning capability is formally defined by the product space of two collections: $\motions$ is the collection of all motions comprised of a full rotation around the turning axis composed with all longitudinal and lateral translations within the machine's workspace, and $\mmfs$ is the collection of all shapes of available tool inserts. Similar definitions are possible for drilling, milling, printing, and other capabilities, based on their DOFs and MMNs. Generally, $\motions$ is a continuum, unless one desires to explicitly model indexed NC movements. On the other hand, $\mmfs$ is always a finite collection, further shortlisted for a particular length scale. Thus we can compute primitives for different instances of the same capability by considering one MMN at-a-time for a range of available motions.

\section{Hybrid Manufacturing Primitives} \label{sec_prim}

\begin{figure*}
	\centering
	\includegraphics[width=0.95\textwidth]{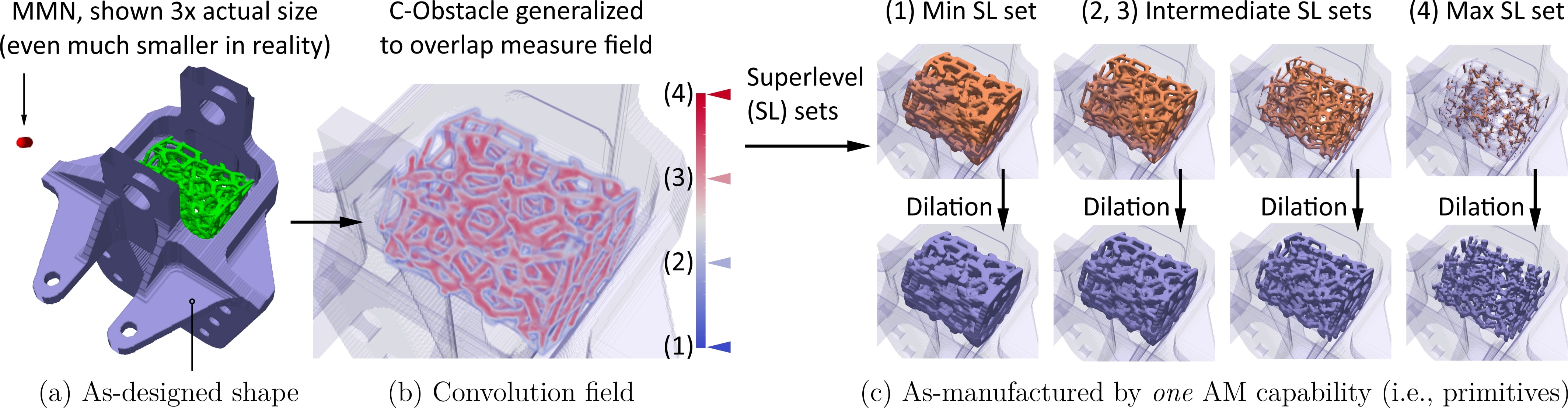} 	
	\caption{AM primitives obtained by varying the allowable overlap measure between a moving MMN and the as-designed shape. The same part as in Fig. \ref{fig_smcost} is used in (a), with the addition of a Voronoi lattice. The field of overlap values shown in (b) is computed as a convolution of indicator functions of $S_\asD$ and $\mmf$, whose superlevel sets are used for the maximal motion $\motion_{\max}$ to obtain the sweeps $S_\asM = \dil(\motion_{\max}^{-1}, \mmf)$.} \label{fig_amprims}
\end{figure*} 

\begin{defn} (Manufacturing Primitive) \label{def_primitive}
	A manufacturing `primitive' is obtained by sweeping any instantiation of a manufacturing capability to define a `region of influence' (ROI) for a manufacturing action.
\end{defn}

Using group morphology \cite{Lysenko2010group} as a proper setting for our formulation, we have $\prim = \dil(\motion^{-1}, \mmf)$ (useful for AM in \eq{eq_action_am}), or, equivalently, $\overline{\prim} = \ero(\motion^{-1}, \overline{\mmf})$ (useful for SM in \eq{eq_action_sm}), since the complement of a morphological dilation (i.e., generalized sweep) is the same as morphological erosion (i.e., generalized unsweep) of the complement \cite{Ilies1999dual,Nelaturi2011configuration}.%
	\footnote{$\sweep(\motion, \mmf) = \dil(\motion^{-1}, \mmf)$ and $\unsweep(\motion, \mmf) = \ero(\motion, \mmf)$ \cite{Lysenko2010group}.}
\begin{equation}
	\dil(\motion, \mmf) = \bigcup_{\pose \in \motion} \pose^{-1} \mmf ~~\text{and}~~ \ero(\motion, \mmf) = \bigcap_{\pose \in \motion} \pose^{-1} \mmf, \label{eq_sweep}
\end{equation}
Each primitive $\prim$ contributes to the making of the part through an action in \eq{eq_action_am} or \eq{eq_action_sm}. The following 4 classes of primitives are of particular interest:
\begin{itemize}
	\item For AM, if $\prim\subseteq S_\asD$ the action is an `under-fill' AM (U-AM); if $\prim \supseteq S_\asD$ it is an `over-fill' AM (O-AM).

	\item For SM, if $\prim \subseteq \overline{S}_\asD$ the action is an `under-cut' SM (U-SM); if $\prim \supseteq \overline{S}_\asD$ it is an `over-cut' SM (O-SM).%
		\footnote{The term `under-cut' is chosen to be consistent with its long-time usage in traditional manufacturing. Unfortunately, under-/over-cut in this context may sound counter-intuitive.}
\end{itemize}
Our goal is to compute a proper collection of manufacturing primitives from a knowledge of the target shape $S_\asD$ that are likely to finish an interchangeable outcome $S_\asM$.

Figure \ref{fig_primitives} illustrates the above 4 types of primitives for a simple 2D shape and three manufacturing capabilities; namely, printing, milling, and turning. Each primitive is obtained {\it independently} by analyzing the target shape without any regard for the other actions/primitives and the order in which they will  be applied to the workpiece.

To obtain a collection $\Prim$ of primitives, each obtained from \eq{eq_sweep}, whose Boolean combinations via \eq{eq_action_am} through \eq{eq_act_compose} can potentially meet a given target specification in a reasonable number of steps, we need a proper criteria to assess whether a given instantiation $(\motion, \mmf)$ whose resulting ROI obtained as  $\prim = \dil(\motion^{-1}, \mmf)$ ``fits best'' to the target design, i.e., the manufacturing action contributes to the overall making of the part to the best of its constrained capability. In our previous work on unimodal process planning (SM alone \cite{Nelaturi2015automatic} and AM alone \cite{Nelaturi2015manufacturability}) we argued that one such criterion comes from set-theoretic maximality of the ROIs. For SM actions, ROIs may be defined in terms of the maximal relative motion $\motion_{\max} \in \motions$ of the manufacturing instrument (i.e., cutting tool) that does not cause collisions with the target shape $S_\asD$. %
For instance, the maximal removable region (MRR) for a SM action is obtained by applying the maximal relative motion of the manufacturing instrument and workpiece that does not result in a collision between the target shape $S_\asD$ and the manufacturing instrument $\inst := (\mmf \cupr \rest)$, including both the MMN (e.g., tool insert) $\mmf$ and other moving components (e.g., spindle) collectively denoted by $\rest$. This motion is given by the complement of {\it configuration space obstacle} ($\conf-$obstacle), also known as the {\it free space} \cite{Lozano-Perez1983spatial}:
\begin{equation}
	\overline{\motion}_{\max} = \obs(S_\asD, \inst) = \big\{ \pose \in \motions ~|~ (S_\asD \capr \pose \inst) \neq \emptyset \big\}. \label{eq_obs_SM}
\end{equation}
Importantly, this approach of defining the ROIs captures {\it accessibility and collision avoidance} for manufacturing instruments of arbitrary shape. If accessibility is taken for granted or analyzed by other means, one can safely let $\rest := \emptyset$ thus $\inst = \mmf$ in subsequent expressions.

The MRR is then obtained as a sweep $\dil(\motion_{\max}^{-1}, \mmf)$ of the MMN (e.g., tool insert) $\mmf$ along the free space. For unimodal SM processes in which any under-cutting would be irreversible, U-SM primitives are not allowable, whereas the largest O-SM primitive is given by the MRR, noting that it is the maximal set (in terms of containment) that is fully contained in the part's complement \cite{Nelaturi2015automatic}.

Similarly, the maximal depositable region (MDR) for an AM action is swept by the maximal relative motion $\motion_{\max} \in \motions$ that keeps the MMN (e.g., a blob of material characterizing the smallest printable shape) fully inside the target shape $S_\asD$, which is the free space for the target shape's complement:
\begin{equation}
	\motion_{\max} = \obs(\overline{S}_\asD, \inst) = \big\{ \pose \in \motions ~|~ (\overline{S}_\asD \capr \pose \inst) \neq \emptyset \big\}. \label{eq_obs_AM}
\end{equation}
MMNs are much smaller for AM compared to SM, thereby making some over-filling tolerable with the understanding that SM post-processing can refine the functional surfaces up to required tolerance. Then $\motion_{\max}$ is redefined by relaxing the condition $(\overline{S}_\asD \capr \pose \inst) \neq \emptyset$ which is equivalent to $\mu[\overline{S}_\asD \capr \pose \inst] = 0$ with $\mu[\overline{S}_\asD \capr \pose \inst] = \lambda \mu[\mmf]$ where $\mu[\cdot]$ denotes the Lebesgue measure (i.e., volume in 3D) and $0 \leq \lambda < 1$ controls a smooth transition from U-AM with the moved MMN fully contained inside the design ($\lambda = 0$) to O-AM with the moved MMN having at least an $\epsilon-$overlap with the design ($\lambda = 1 - \epsilon$) where $\epsilon \to 0^+$ \cite{Nelaturi2015representation} (Fig. \ref{fig_amprims}). Over-filling will allow manufacturing features of $S_\asD$ that are smaller than the characteristic size of MMN, but it comes at the expense of thickening the part, which may need SM finishing of surfaces (e.g., for fit/assembly).

As illustrated in Fig. \ref{fig_primitives}, U-AM/O-SM actions alone---which are the only options in unimodal sequences---cannot produce thin walls/slot or sharp/dull corners, respectively, unless they are replaced with O-AM/U-SM actions followed by other O-SM/U-AM actions in a multimodal sequence. In practice, O-AM primitives are needed to leave some allowance for finishing by SM as a typical fabrication scenario \cite{Yamazaki2016development,Du2016novel}, and U-SM primitives are useful to carve out cracks for later refilling by AM as a repair mechanism \cite{Ren2007part,Liou2007applications}. Thus we use the MDR/MRR formulation to compute U-AM/O-SM primitives:
\begin{align}
	\prim_{\max} &= \dil(\overline{\obs(S_\asD, \inst)}^{-1}, \mmf), \quad\text{for U-AM actions}, \\
	\overline{\prim}_{\max} &= \ero(\overline{\obs(S_\asD, \inst)}^{-1}, \overline{\mmf}), \quad\text{for O-SM actions},
\end{align}
which can be simplified by exploiting relationships between dilation/erosion and Minkowski products/quotients \cite{Lysenko2010group}. Here we present the simplest case for a $3-$axis machine with translational DOF operating at a fixed build/fixturing orientation per primitive (i.e., $\motions \subset \powerset(\euc)$):
\begin{align}
	\prim_{\max} &= (S_\asD \ominus (-\inst)) \oplus \mmf, \quad\text{for MRR U-AM}, \label{eq_opening_UAM_trans} \\
	\overline{\prim}_{\max} &= (S_\asD \oplus (-\inst)) \ominus \mmf, \quad\text{for MDR O-SM}, \label{eq_closing_OSM_trans}
\end{align}
where $\oplus, \ominus: (\rset \times \rset) \to \rset$ are regularized Minkowski sum/difference,  and $-\inst = \{-\bx ~|~ \bx \in \inst\}$ denotes a reflection of the potentially non-symmetric $D = (\mmf \cupr \rest)$. Here, $\rest$ refers to the inactive moving components (relative to the workpiece) of the manufacturing instrument. These further reduce to the morphological opening/closing, respectively, of the target part $S_\asD$ with the `structuring element' $\mmf$ \cite{Serra1983image} when it is safe to assume $\rest = \emptyset$ thus $\inst = \mmf$.

Unlike maximal U-AM/O-SM, minimal O-AM/U-SM primitives are non-unique in the partial ordering of sets. For example, there are infinitely many ways to generate the simple O-AM/U-SM primitives illustrated in Fig. \ref{fig_corner} with a ``bulged'' corner by adding configurations that populate copies of the MMN at the corner such that ROI will then cover the originally under-filled or over-cut regions---e.g., using a few additional weld-spots for O-AM or drill-holes for U-SM. Many so-obtained primitives will be minimal (despite being non-unique), i.e., not  reducible to another sweepable primitive fully contained inside them. Although they can be perceived as {\it economic} choices, they may not be suitable when continuous motions are preferred to applying the additional deposition/removal at discrete spots. A {\it conservative} alternative would be to sweep the MMN along the entire shape at allowable orientations. For the simplest example of $3-$axis translational DOF at a fixed build/fixturing orientation, the primitives are obtained by:
\begin{align}
	\prim_{\cons} &= (S_\asD \ominus (-\rest)) \oplus \mmf, \quad\text{for cons. O-AM}, \label{eq_opening_UAM_trans} \\
	\overline{\prim}_{\cons} &= (S_\asD \oplus (-\rest)) \ominus \mmf, \quad\text{for cons. U-SM}, \label{eq_closing_OSM_trans}
\end{align}
which reduce to global thickening/shrinking (i.e., generalized $\pm-$offset) of the target design by MMN when is safe to assume $\rest = \emptyset$ thus $\inst = \mmf$. 

\begin{figure}
	\centering
	\includegraphics[width=0.45\textwidth]{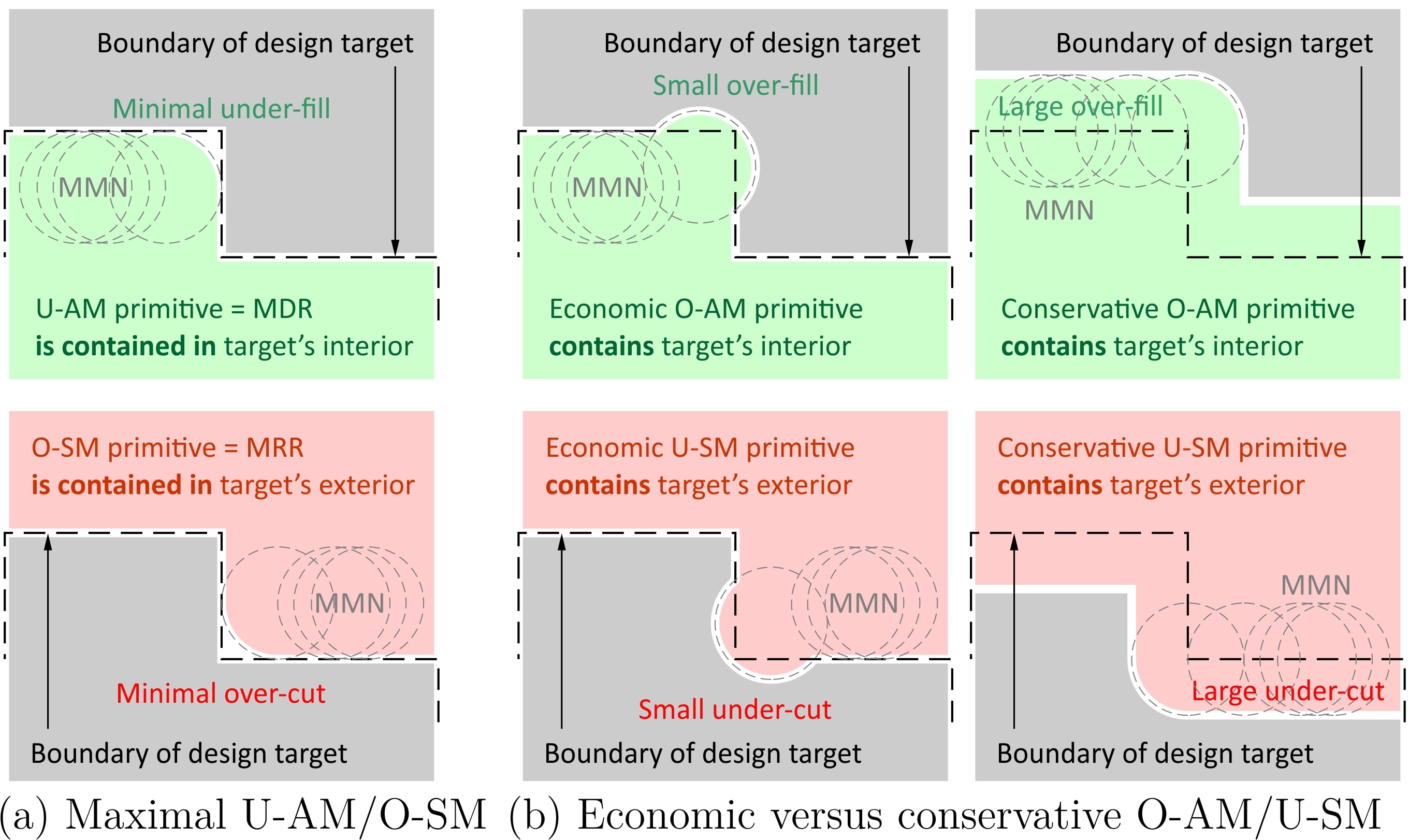}
	\caption{U-AM/O-SM primitives (the only options in unimodal AM/SM), if modeled by sweeps of a round MMN, cannot produce both sharp and dull corners. In a multimodal sequence, this is possible by combining U-AM+U-SM and O-AM+O-SM.} \label{fig_corner}
\end{figure}

It has been shown \cite{Lysenko2010group,Nelaturi2012rapid,Behandish2016analytic} that morphological operations over $\conf$ may be computed efficiently using the correspondence with group convolutions. Using this correspondence, the Minkowski sum/difference of sets may be calculated efficiently as min/max superlevel sets of convolutions of their indicator functions, while other levelsets in between correspond to $\lambda-$overlap measures (Fig. \ref{fig_amprims}) discussed earlier with regards to relaxing \eq{eq_obs_AM}. Convolutions, in turn, may be implemented via fast Fourier transforms (FFT) as a result of the well-known {\it convolution theorem}, thereby making the computations run in $O(n \log n)$ time for a sample size of $n$. This approach is agnostic of the geometric complexity of part or tool profiles and may be implemented rapidly on GPUs. The 3D primitives in Section \ref{sec_results} have been computed using this approach.

In many cases one needs a localized offset around specific features. In those cases, it is unlikely for a generic set-theoretic recipe to provide suitable primitives and a more localized approach may be preferable. The primitives could be obtained by breaking the shape $S_\asD$ into a number of components---e.g., using feature-recognition or domain decomposition  methods \cite{Woo1982feature}---and applying the morphological operations mentioned above to each component separately. Note that neither the components nor their primitives are required to be mutually disjoint, since our model of successive execution (Definition \ref{def_process}) allows for overlapping primitives and the decomposition described in Section \ref{sec_decomp} will take care of mutual intersections.

It is important to emphasize that the validity of the manufacturability test and process planning formulated in the following sections are completely independent of the choice of primitives. {\it The following results are valid for manufacturing primitives of arbitrary geometric shapes}, no matter how they were generated, automatically as described above or manually by an expert user, whether they respect the machine DOF and MMN constraints, or if they are maximal/minimal with respect to any metric.

\section{Hybrid Manufacturability Analysis} \label{sec_decomp}

\begin{figure*}
	\centering
	\includegraphics[width=0.95\textwidth]{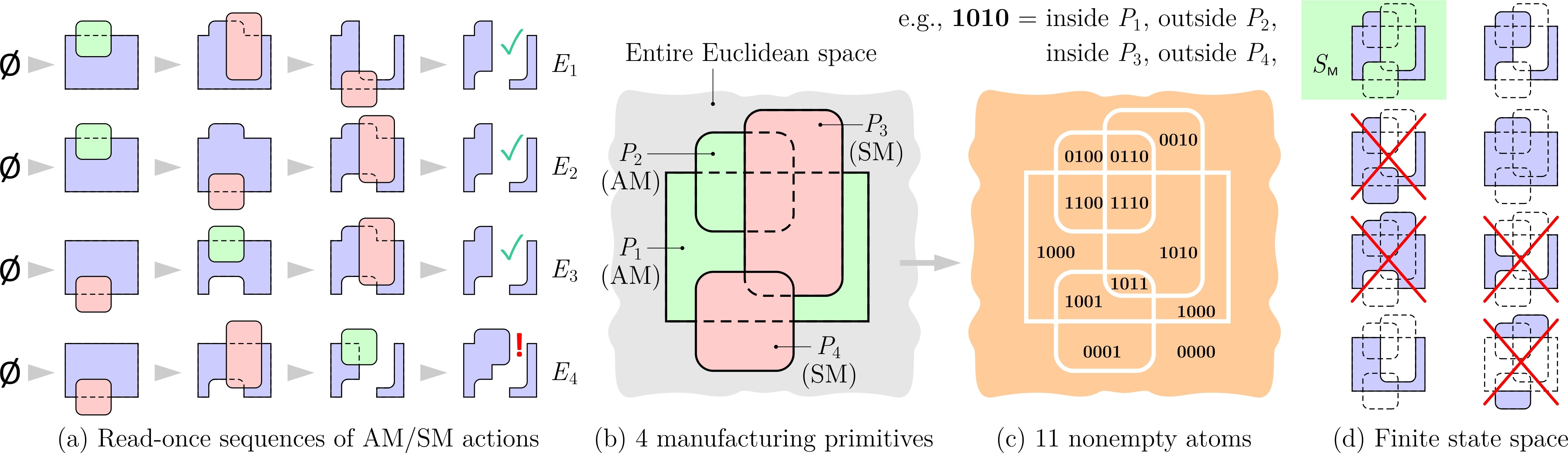}
	\caption{Unlike unimodal processes, multimodal processes are not permutative (a). However, every manufacturing solid using a given collection of the primitives (b) is the union of some subcollection of their canonical intersection terms (c). The converse is not true (d).} \label{fig_decomp}
\end{figure*}

The as-planned view is expressed as a finite Boolean expression (FBE) $\expr[\Prim]$ in terms of elements of $\Prim$, whose evaluation produces the as-manufactured view $S_\asM$. Not every FBE is a `valid' plan---i.e., one that has a meaningful physical realization as a sequence of AM/SM actions. The validity conditions are as follows:

\begin{defn} (Manufacturing Plan) \label{def_plan}
	A manufacturing `plan' $\expr[\Prim]$ is a FBE over $\Prim$ that satisfies the following:
	\begin{description} [style=unboxed,leftmargin=0cm]
		\item [C1] the FBE is `anti-balanced' in terms of its CSG-tree representation, meaning that every intermediate node in the tree has at least one leaf child:
		\begin{equation}
			\centering
			\begin{tikzcd}[row sep=0.05em, column sep=small]
				\Prim^\uparrow \arrow[rd, "\text{pick}"] & ~ & \Prim^\ast \arrow[rd, "\text{pick}"] & ~ & \Prim^\ast \arrow[rd, "\text{pick}"] \\
				~ & \bigcirc \arrow[rd] & ~ & \bigcirc \arrow[rd] & ~ & \cdots \arrow[rd] \\
				\Prim^\ast \arrow[ru, "\text{pick}"'] & ~ & \bigcirc \arrow[ru] & ~ & \bigcirc \arrow[ru] & ~ & S_\asM \\
				~ & \Prim^\ast \arrow[ru, "\text{pick}"'] & ~ & \Prim^\ast \arrow[ru, "\text{pick}"']
			\end{tikzcd} \label{eq_csg}
		\end{equation}
		where $\bigcirc$ embodies a Boolean operation with or without set complement. For example, $(\prim_1 \cupr \prim_2) \capr \prim_3$ is anti-balanced, but $(\prim_1 \cupr \prim_2) \capr (\prim_3 \cupr \prim_4)$ is not.

	\item [C2] AM/SM primitives appear in disjunctive/conjunctive and positive/negative forms, respectively, i.e., respect the AM/SM action semantics in \eq{eq_action_am} or \eq{eq_action_sm}---AM actions always appear as $((\cdot) \cupr \prim)$ while SM actions appear as $((\cdot) \cap \overline{P})$.
	
	\item [C3] The first primitive must be additive, to model initiating the process by actual material deposition (e.g., 3D printing) or selection of a raw/bar stock.
	
	\item [C4] The subsequent primitives must be picked from $\Prim^\ast = (\Prim - \Prim_\raw)$ (excluding raw/bar stocks), to model modifications by actual material deposition/removal.
	\end{description}
\end{defn}

The above validity conditions define the search space of physically meaningful HM plans, of which a smaller subset of plans produce the desired outcome. A given part $S_\asD$ is `manufacturable' via a predefined collection of AM/SM primitives, if there exists at least one FBE that satisfies \textbf{C1--4} and whose outcome $S_\asM$ is interchangeable with $S_\asD$.

In the special case of unimodal processes (all-AM or all-SM) with arbitrary ROIs, or multimodal processes with completely separated ROIs, all plans that consume every primitive in $\Prim$ {\it have only one possible outcome}; namely, $S_\asM = \bigcup_{1 \leq i \leq n} \prim_i$ (all-AM), $S_\asM = \bigcap_{1 \leq i \leq n} \overline{\prim}_i$ (all-SM), or their combination in arbitrary orders (disjoint ROIs). {\it Unimodal manufacturability is decidable prior to planning,} i.e., the test for interchangeability of $\expr[\Prim]$'s output shape $S_\asM$ depends only on $\Prim$ and is the same for all expressions with different orders of depositing/removing the primitives. The difference is in the cost of manufacturing as well as physical and operational constraints that are not accounted for in the geometric model, and will be evaluated during planning. This decoupling of testing for manufacturability (i.e., if there is any chance to find a plan) from the task of shortlisting the most cost-optimal and practical plans is key to making unimodal AM/SM planning with arbitrary primitives tractable.
This approach can be extended to multimodal HM planning for a predefined structure with a few (i.e., $O(1)$) unimodal subsequences that are permutative amongst themselves with respect to the geometric outcome. %
However, we would like to enable HM for general sequences without such restraints, in which the lack of permutativity leads to a combinatorial blow-up of the search space without any a priori guarantees of manufacturability. In this section we show that early tests for manufacturability can be obtained without making restrictive assumptions about the HM sequence.

The space of all finite Boolean functions (FBF) applied to the elements of $\Prim$ forms a so-called free/finite Boolean algebra (FBA). The primitives are the FBA's `generators', i.e., a distinguished subset of the algebra from which every other element can be produced using a finite number of Boolean operations. Importantly, every FBA has another set of distinguished elements called canonical intersection terms (i.e., `atoms') from which every other element of the algebra can be disjointly assembled. This results from logical and set-theoretic principles, and is true {\it irrespective of the geometric or topological characteristics of the space in which the elements of algebra are embedded,} reinforcing the approach's validity for all dimensions (e.g., 2D and 3D) and arbitrarily complex shapes.

\begin{defn} (Manufacturing Atoms)
	Given a collection of manufacturing primitives $\Prim = \{\prim_1, \prim_2, \ldots, \prim_n\}$ as the FBA generators, the manufacturing atoms are the FBA's canonical intersection terms $\Atom = \{ \atom_1, \atom_2, \ldots, \atom_m\}$:
	\begin{equation}
		\atom_j = \bigcap_{1 \leq i \leq n} Q_{i,j} = (Q_{1, j} \capr Q_{2, j} \capr \cdots \capr Q_{n, j}), \label{eq_atom}
	\end{equation}
	where $Q_{i, j}$ is either $\prim_i$ or $\overline{\prim}_i$, depending on $j = 1, 2, \ldots, m$, e.g., on the $i$\th digit of its binary encoding. There are $2^n$ possible combinations but $m = |\Atom| \leq 2^n$ since all empty intersections produce the same `empty atom', which proves to be critical in subsequent discussions.
\end{defn}
	
The atoms are mutually disjoint cells that partition Euclidean space into {\it the smallest deposition{\rm/}removal units that an arbitrary combination of actions can produce}---not in isolation, but with implicit co-appearance conditions imposed by the composition of primitives from atoms. All points in each partition have identical membership classification against every primitive. Unlike a primitive, an atom does not represent physical significance other than which AM/SM primitives can potentially contribute to its existence (or lack thereof) at any state transition along the process. All atoms that comprise a given primitive need to co-appear/disappear when an action is executed.

The atomic decomposition is computed directly from the {\it unordered} collection of primitives. Once obtained, it provides a volumetric spatial enumeration scheme to represent all manufacturable parts by enumerating a finite subset of atoms, i.e., specifying which atoms are in/out. Each atom is either included or excluded in its entirety (i.e., no partial inclusions) in every manufacturable part, regardless of what plan makes it. As the name suggests, {\it an atom cannot be split by any manufacturing plan;} a fact that enables HM manufacturability analysis prior to the costlier task of process planning.

\begin{theo} (Necessary Condition for Manufacturability) \label{theo_dnf}
	Every manufacturable shape by a predefined collection of primitives is a union of a subcollection of nonempty atoms:
	\begin{equation}
		S_\asM = \bigcup\Atom \langle \bJ_\asM \rangle = \bigcup_{j \in \bJ_\asM} \atom_j, ~\text{for some}~ \bJ_\asM \subseteq \{1, 2, \ldots, m\}. \label{eq_subcol}
	\end{equation}
	The notation $\Atom \langle \bJ \rangle = \{ \atom_j ~|~ j \in \bJ\}$ is used hereafter to refer to an atomic enumeration of any solid in the FBA.
\end{theo}

\begin{proof}
	It is straightforward to prove by induction over \eq{eq_csg} that every pair of points that are in the same atom, i.e., classify the same way against all primitives, have to classify the same way against $S_\asM = \expr[\Prim]$ as well. Note that this is true {\it regardless of the order of appearance of primitives.} Thus every outcome of an FBE has to include either all or none of the points in any given atom. This includes manufacturable shapes that are (by definition) outcomes of manufacturing plans (i.e., a special form of FBE).
\end{proof}

Theorem \ref{theo_dnf} gives a {\it necessary but insufficient} condition for manufacturability. The FBA is a superset of all manufacturable solids since every manufacturable solid admits an atomic cover, but the converse is not necessarily true. Given an as-designed target $S_\asD$ and criteria for part interchangeability (e.g., GD\&T \cite{ASME2009dimensioning}), once the primitives are computed using any method of choice and their atomic decomposition is determined, the next step before process planning is to test if there exists a subcollection $\Atom \langle \bJ_\asM \rangle$ of nonempty atoms as in \eq{eq_subcol} whose union passes the interchangeability test. In the most strict scenario where precise equality is required (i.e., $S_\asD = S_\asM$), one needs to test if all atoms in $\Atom \langle \bJ_\asM \rangle$ are either fully inside or fully outside the target design, i.e., $(\atom_j \capr \overline{S}_\asD) = \emptyset$ or $(\atom_j \capr S_\asD) = \emptyset$ for all $j \in \bJ_\asM$. If there exists even one atom that partially collides with $\overline{S}_\asD$ and $S_\asD$, it becomes immediately evident that no Boolean expression $\expr[\Prim]$ can precisely produce $S_\asD = S_\asM$. When the equality is relaxed and slight deviations (e.g., within some tolerance zones) are allowed, the said two intersections must be either empty or fully contained inside the tolerance zone. ``Tolerant'' atomic decomposition is a subject of ongoing research.

\section{Hybrid Process Plan Enumeration} \label{sec_hmplans}

The subsequent analysis resides in the realm of Boolean logic and is not affected by geometric complexities. Using this to our advantage, we use simple overlapping boxes (akin to Venn diagrams) like the ones in Fig. \ref{fig_decomp} to represent {\it generic} primitives for the remaining discussion, before demonstrating applicability to real 3D examples for HM in Section \ref{sec_results}. These schematics will serve to present ideas that are difficult to illustrate on real 3D primitives.

It is convenient to use bit-strings to represent the atom indices as $j = \prod_{1 \leq i \leq n} 2^{n-i} b_i$, where the bit $b_i = 1$ or $0$ determines how $\atom_j$ classifies against $\prim_i$, i.e., completely inside or completely outside. For example, given the four primitives of Fig. \ref{fig_decomp} (b), there are $2^4 = 16$ canonical intersections. $11$ of them create nonempty atoms in this case, indexed accordingly in Fig. \ref{fig_decomp} (c). Here are some of them to illustrate the indexing scheme:
\begin{align}
	A_{0100} &= (\overline{\prim}_1 \capr \prim_2 \capr \overline{\prim}_3 \capr \overline{\prim}_4), \label{eq_atom1} \\
	A_{1000} &= (\prim_1 \capr \overline{\prim}_2 \capr \overline{\prim}_3 \capr \overline{\prim}_4), \label{eq_atom2} \\
	A_{1100} &= (\prim_1 \capr \prim_2 \capr \overline{\prim}_3 \capr \overline{\prim}_4), \qquad\ldots \label{eq_atom3}
\end{align}
One can think of the indices as the outcome of membership classifications of any representative point inside each atom against all primitives in some fixed order.

By substituting the conjunctive expression in \eq{eq_atom} for the subcollection of nonempty atoms $\Atom \langle \bJ_\asM \rangle$ in the disjunctive expression in \eq{eq_subcol}, we obtain a disjunctive normal form (DNF) for the manufactured outcome. For example, the 3 nonempty atoms in \eq{eq_atom1} through \eq{eq_atom3} collectively represent a decomposition of a potentially manufacturable solid shown in Fig. \ref{fig_decomp} (d) (top-left) decomposed as:
\begin{equation}
	S_\asM = ~(A_{0100} \cupr A_{1000} \cupr A_{1100}). \label{eq_atomic}
\end{equation}
Substituting \eq{eq_atom1} through \eq{eq_atom3} in \eq{eq_atomic} yields a DNF denoted hereafter by $\expr^\dag[\Prim]$. It has no physical meaning by itself in terms of sequence of AM/SM actions. It rather prescribes the disjoint atomic pieces (i.e., conjunctive clauses, or `implicants') that fully cover the interior of a manufactured part and how each piece classifies against the AM/SM primitives. The next question is whether the DNF {\it can be rearranged by purely symbolic manipulation to a valid manufacturing plan} qualified via conditions \textbf{C1--4} of Definition \ref{def_plan}, that is, an anti-balanced CSG-tree that is respectful of AM/SM modalities.
The caveat is that symbolically different DNFs can be formed, corresponding to different FBFs, to produce the same outcome $S_\asM$ solely because of the existence of {\it empty atoms}. While the `minimal' DNF is formed by the fewest conjunctive clauses corresponding to nonempty atoms, any combination of the clause(s) that correspond to empty atoms can also be included to formulate an equally valid `enriched' DNF. All (minimal or enriched) DNFs produce the same manufactured shape, nonetheless represent {\it logically different} FBFs. If one desires to reframe the problem in purely symbolic language, one has to incorporate all shape-related information into logical reasoning by means of DNF enrichment.

\begin{theo} (Sufficient Condition for Manufacturability) \label{theo_pro}
	The union of a given subcollection of nonempty atoms is manufacturable if at least one of its DNFs (potentially enriched with empty atoms) is logically equivalent to an FBE that satisfies \textbf{C1--4} of Definition \ref{def_plan}.
\end{theo}

\begin{proof}
	The equivalent plan (by definition) produces a manufacturable shape $S_\asM = \expr[\Prim]$ and $\expr[\Prim] = \bigcup \Atom \langle \bJ_\asM \rangle$ for $\bJ_\asM$ containing a subset of indices of nonempty atoms (Theorem \ref{theo_dnf}). There always exists a DNF $\expr^\dag[\Prim]$ that is equivalent to $\expr[\Prim]$, thus $S_\asM = \expr^\dag[\Prim]$. The implicants of the DNF correspond to some subcollection of atoms with indices $\bJ^\ast$ (including indices that may create empty atoms) or their simplifications (by conjunction) into simpler implicants, i.e., $\expr^\dag[\Prim] = \bigcup \Atom \langle \bJ^\ast \rangle$ meaning that $\Atom \langle \bJ_\asM \rangle$ and $\Atom \langle \bJ^\ast \rangle$ can differ only by the empty atom $(\Atom \langle \bJ^\ast \rangle - \Atom \langle \bJ_\asM \rangle) \subseteq \{ \emptyset\}$. This means that $(\bJ^\ast - \bJ_\asM)$ can only include indices corresponding to empty canonical intersections. The DNF $\expr^\dag[\Prim]$ differs from minimal DNF by uniting implicants corresponding to $(\bJ^\ast - \bJ_\asM)$, i.e., is its enrichment.
\end{proof}

\begin{figure}
	\centering
	\includegraphics[width=0.48\textwidth]{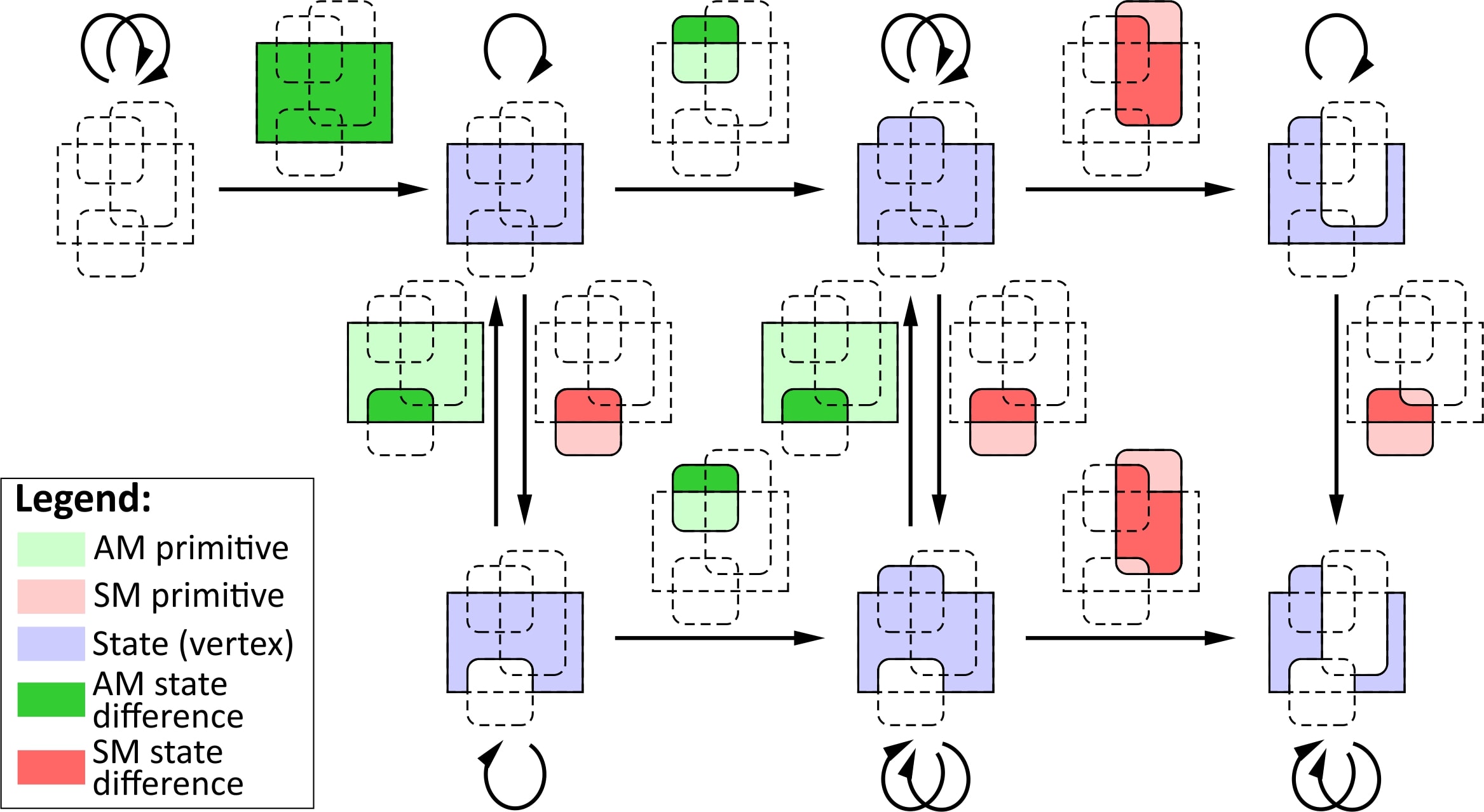}
	\caption{State transitions through manufacturing plans are conceptualized by turning atoms on and off as a result of union/intersection with AM/SM primitives, respectively. The cost along each path is a linear function of state difference volumes and action cost factors.} \label{fig_graph}
\end{figure}

For instance, the minimal DNF corresponding to \eq{eq_atomic} for the part in Fig. \ref{fig_decomp} (d) (top-left) can be united with the implicants corresponding to any subcollection of the $5$ empty atoms $A_{0011}$, $A_{0101}$, $A_{0111}$, $A_{1101}$, and $A_{1111}$ to obtain a different DNF, each of which may or may not be rearrangeable into valid plans. %
For example, the minimal DNF itself is equivalent to the top two plans in Fig. \ref{fig_decomp} (a):
\begin{align}
	\expr_1[\prim_1, \prim_2, \prim_3, \prim_4] &= ( ( ( \prim_1 \cupr \prim_2 ) \capr \overline{\prim}_3 ) \capr \overline{\prim}_4 ), \label{eq_examp1} \\
	\expr_2[\prim_1, \prim_2, \prim_3, \prim_4] &= ( ( ( \prim_1 \cupr \prim_2 ) \capr \overline{\prim}_4 ) \capr \overline{\prim}_3 ), \label{eq_examp2}
\end{align} 
as well as the other two plans (not shown) obtained by permuting $\prim_1$ and $\prim_2$. However, the third plan in Fig. \ref{fig_decomp} (a) is equivalent to a different DNF obtained by uniting two more implicants corresponding to $A_{0101}$ and $A_{1101}$:
\begin{equation} 
	\expr_3[\prim_1, \prim_2, \prim_3, \prim_4] = ( ( ( \prim_1 \capr \overline{\prim}_4 ) \cupr \prim_2 ) \capr \overline{\prim}_3 ). \label{eq_examp3}
\end{equation}
$\expr_1$ and $\expr_2$ are {\it logically} equivalent to each other for all shapes and configurations of primitives, but not to $\expr_3$ since the latter represents a different FBF in general. But they are {\it conditionally} equivalent, i.e., produce the same outcome for the shapes and configurations in Fig \ref{fig_decomp} (b) that results in $(P_2 \capr P_4) = \emptyset$ leading to $\atom_{\ast 1 \ast 1} = \emptyset$. The symbol ``$\ast$'' encodes both $0$ and $1$, thus $\atom_{\ast1 \ast 1}$ is the short form for $4$ atoms that collectively form an implicant.%
	\footnote{In this case, $\atom_{\ast 101} = \emptyset$ are sufficient, since even if $\atom_{0011}$, $\atom_{0111}$, and $\atom_{1111}$ are nonempty, the appearance of conjunction with $\overline{\prim}_3$ after disjunction with $\prim_2$ in both FBEs excludes them. This shows how nontrivial it can get to reason about conditional equivalence.}
Lastly, the fourth plan in Fig. \ref{fig_decomp} (a) produces a completely different outcome for the same primitive set:
\begin{equation}
	\expr_4[\prim_1, \prim_2, \prim_3, \prim_4] = ( ( ( \prim_1 \capr \overline{\prim}_4 ) \capr \overline{\prim}_3 ) \cupr \prim_2 ). \label{eq_examp4}
\end{equation}

\begin{figure}
	\centering
	\includegraphics[width=0.48\textwidth]{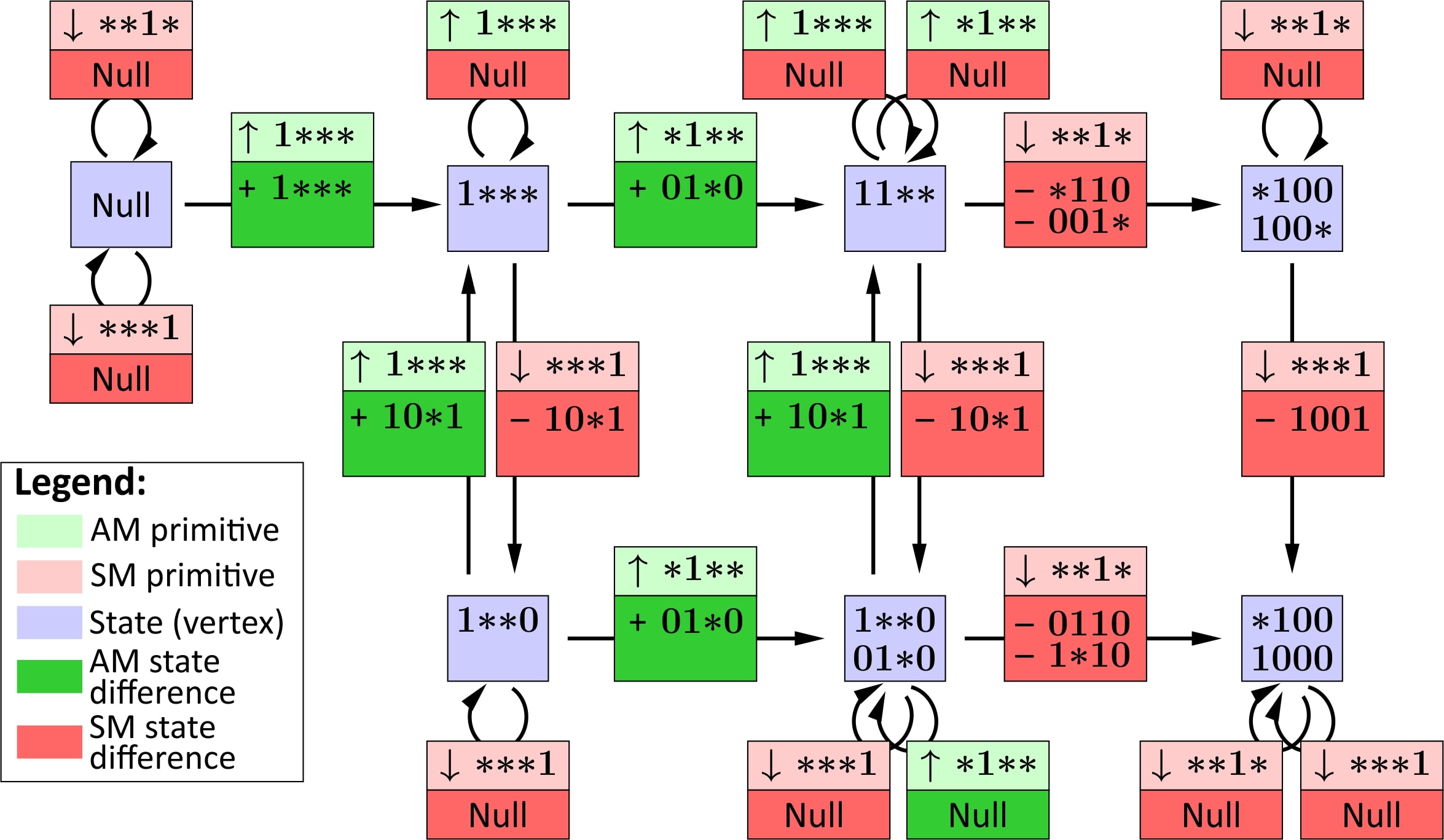}
	\caption{The planner works with bitwise atomic representation of states (on vertices) and state differences (on edges), and does not need to understand the geometric interpretation as in Fig. \ref{fig_graph}.} \label{fig_planner}
\end{figure}

Theorems \ref{theo_dnf} and \ref{theo_pro} give a {\it necessary and sufficient} condition for manufacturability. The challenge with the latter is that unlike the nonempty atoms, the empty pairwise intersections can be numerous---i.e., $(2^n - m) = O(2^n)$ if $m \ll 2^n$. Although atomic decomposition of $S_\asM$ and logical manipulation of the minimal DNF are tractable, iterating over all of the enriched DNFs in search of valid plans can lead to exponential complexity in the worst case. Often it makes sense to use the weak test (Theorem \ref{theo_dnf}) but omit the strong test (Theorem \ref{theo_pro}) to gain computational advantage prior to planning. In any case, direct process planning using a standard search algorithm (e.g., A$^\star$ best-first search \cite{Korf2010algorithms}) may be used to generate different permutations of actions and test if any of them is at least conditionally equivalent to the minimal DNF, i.e., evaluates to the same outcome with respect to the implicants corresponding to nonempty atoms only.

Figure \ref{fig_graph} and \ref{fig_planner} illustrate the state transitions of the workpiece along the 3 valid plans represented by FBEs in \eq{eq_examp1} through \eq{eq_examp3}. Note that the planning operations are represented in purely symbolic terms (Fig. \ref{fig_planner}).

\section{Design for Hybrid Manufacturing} \label{sec_design}

\begin{figure}
	\centering
	\includegraphics[width=0.45\textwidth]{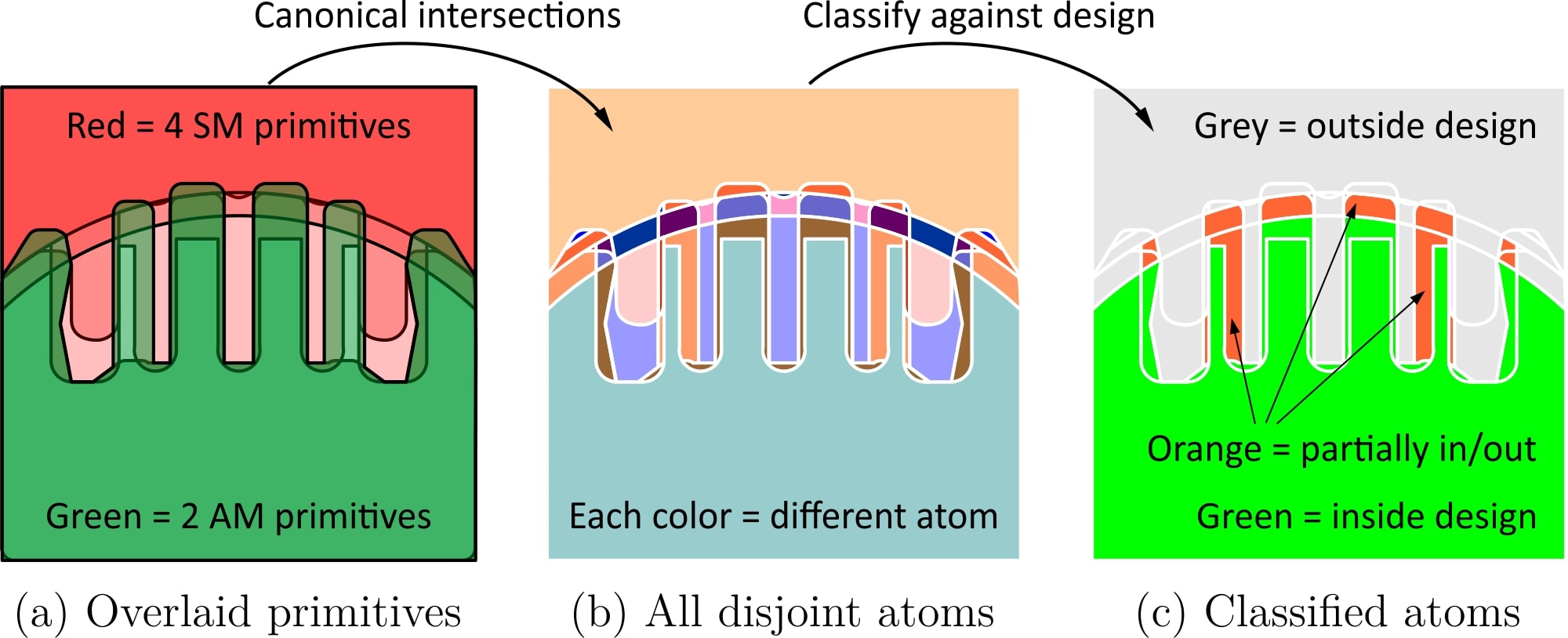}
	\caption{The 6 primitives of Fig. \ref{fig_primitives} are placed in layers on top of each other in (a) to reveal the atomic units that enumerate every sequence of their Boolean combinations in (b). Classifying atoms against the target yields an early test for manufacturability in (c).} \label{fig_atoms}
\end{figure}

The partially colliding atoms discussed in Section \ref{sec_decomp} provide targeted information to guide the choice of additional primitives that could potentially make the part manufacturable. For example, if an atom $\atom$ partially collides with a nominal shape $S_\asD$ and provides two nonempty `subatoms' $\atom_\mathrm{i} := (\atom \capr S_\asD)$ and $\atom_\mathrm{o} :=(\atom \capr \overline{S}_\asD)$, it becomes immediately obvious that a new primitive is needed to properly split this atom into its subatomic pieces. Every new primitive $\prim_{n+1}$ that is added to $\Prim$ will split all original atoms that it partially intersects into new nonempty atoms $\atom_\mathrm{i}' := (\atom \capr \prim_{n+1})$ and $\atom_\mathrm{o}' = (\atom \capr \overline{\prim}_{n+1})$. To eliminate the non-manufacturability issue due to $\atom$, one has to ``design'' the primitive such that $\atom_\mathrm{i} = \atom_\mathrm{i}'$ and $\atom_\mathrm{o} = \atom_\mathrm{o}'$.

Starting from a small initial collection of primitives ($|\Prim| = O(1)$) that produces an initial {\it coarse-grained} atomic decomposition, more primitives can be added iteratively to obtain an increasingly {\it fine-grained} decomposition until all partially colliding atoms are split into subatoms that are either fully inside or fully outside the design target. Figure \ref{fig_atoms} exemplifies how atoms are classified against the target and how they impose constraints on the design of new primitives that can split the ones contributing to non-manufacturability. More research is needed on ``iterative fine-graining'' for HM process planning.

In addition to process planning to map a given designed form to a sequence of actions, this is a first step towards systematic design for hybrid manufacturing (DfHM). By enumerating the space of all potentially hybrid manufacturable shapes for given primitives, the canonical decomposition provides a modeling paradigm and representation scheme to formulate the DfHM problem as one of designing primitives whose induced partitioning of space enumerates designs that satisfy a given function. In contrast to the common modeling tools and representation schemes in current use, the cellular enumeration in terms of HM atoms constructs a design space that, upon checking necessary and sufficient conditions of Theorems \ref{theo_dnf} and \ref{theo_pro}, {\it produces a priori guarantees of manufacturability as a powerful alternative to a posteriori verification and design feedback.}

\section{Results and Discussion} \label{sec_results}

\begin{figure*} 
	\centering
	\includegraphics[width=0.95\textwidth]{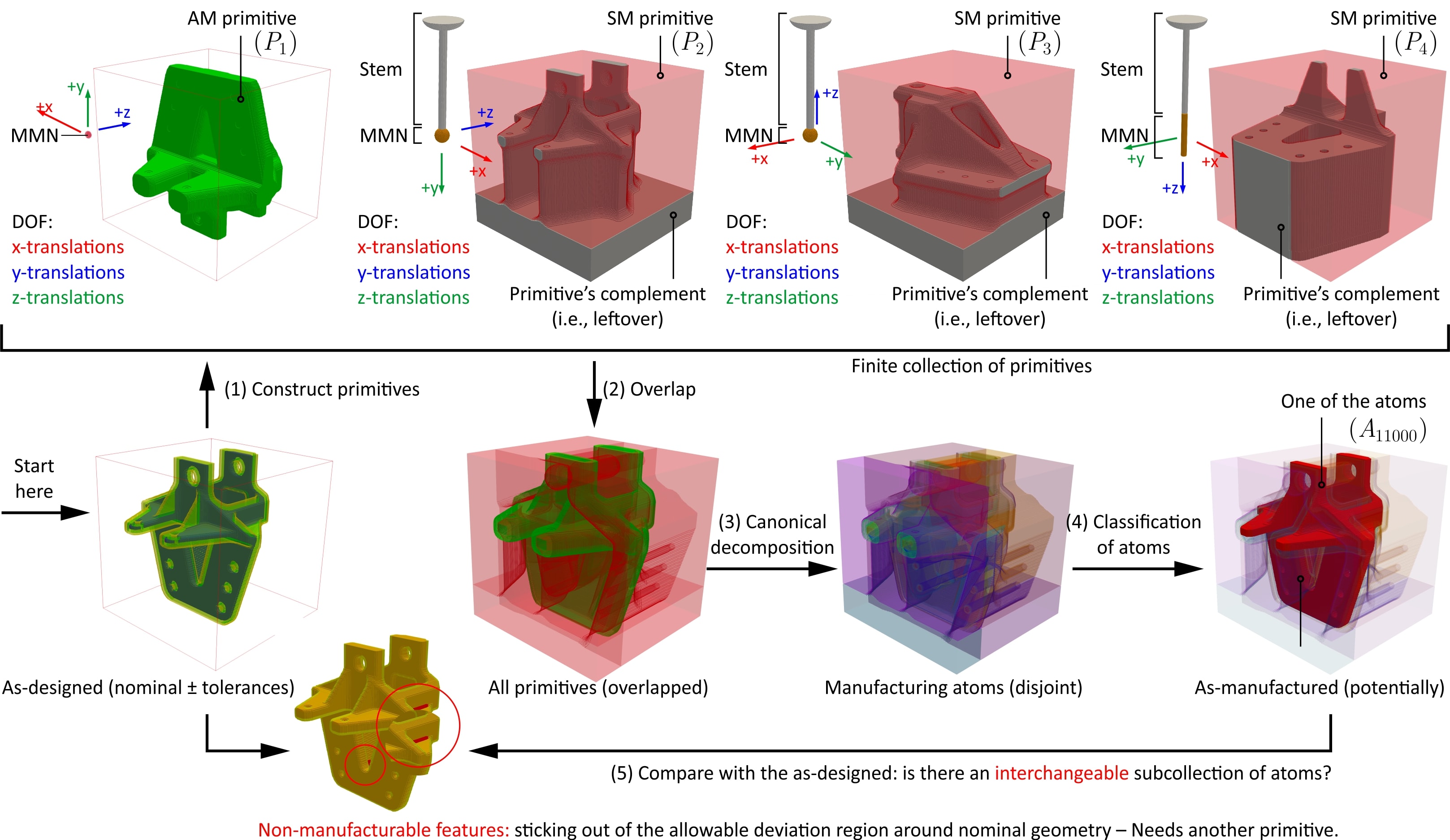}
	\caption{2 AM (including raw stock, not shown) and 3 SM primitives are constructed for a $3-$axis machine with a few HM capabilities (top). The primitives are overlapped to construct an atomic decomposition whose atoms are checked against the as-designed shape to discover an interchangeable as-manufactured shape (bottom) or deviations that required to be fixed by adding more primitives to split them.} \label{fig_prims3D}
\end{figure*}

We first revisit HM analysis of the bracket example from Section \ref{sec_intro} on a 3-axis machine with translational DOF. As illustrated in Fig. \ref{fig_smcost}, SM alone leads to substantial waste of material removed from the raw stock.
To alleviate this, we construct an O-AM primitive using the conservative formula in \eq{eq_opening_UAM_trans}. We also construct 3 O-SM primitives by applying the MRR formula in \eq{eq_opening_UAM_trans} at 3 different fixture setups, using ballmill and endmill tools. For the sake of simplicity, we are not considering issues pertaining to optimizing build orientation, scaffolding, or fixturing. Figure \ref{fig_prims3D} illustrates the workflow including the atomic decomposition of the 5 overlapped primitives.
To make up a simple test of interchangeability, let us pick a global offset tolerance of 3\% of the bounding box edge length, which is smaller than tool diameters (5$-$10\%) and (a largely exaggerated) minimum printable feature size (3$-$5\%). As a result, it is expected that some of the sharp/dull corners will not be manufacturable. This is confirmed by the classification of the atoms against the design specification. A single atom $A_{11000}$ prevails as one that slightly sticks out of the tolerance zone. The issue would disappear by either augmenting the AM/SM capabilities, using higher-resolution deposition heads or thinner/longer milling tools; or altering the design using larger nominal fillets or relaxed tolerances around surfaces of little functional significance.
To do the former automatically, the colliding atom $A_{11000}$ is marked for splitting into subatoms, and new primitives are designed using existing capabilities with proper DOF or MMN to locally target the marked atoms. Importantly, this does not require any design parameterization or assumptions on feature semantics---e.g., what a ``fillet'' means and how its radius is changed to fix the issue. Adopting a DfHM approach for the latter constrains the designer to respect the existing capabilities in the first place. By exploring the design space of parts that are enumerable by atomic decomposition of MMN-sweepable primitives with existing capabilities, only manufacturable fillets could be morphologically generated.

\begin{figure*}
	\centering \includegraphics[width=0.95\textwidth]{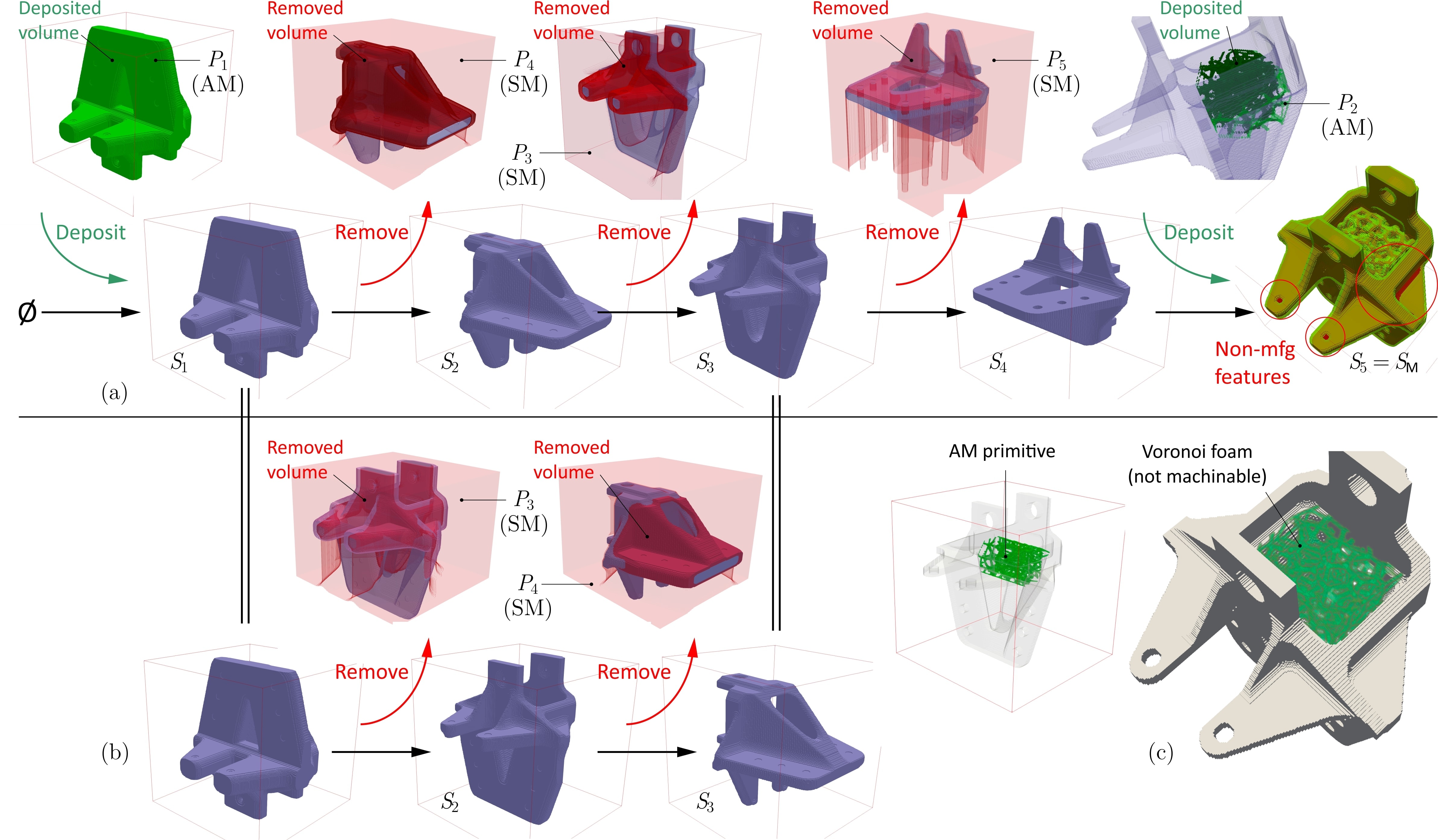}
	\caption{The cost-optimized plans (a, b) to construct an as-manufactured part with minimal deviation from the as-designed part in (c). The detour in (b) shows a permutation of a unimodal subsequence which preserves the shape, since $((S_1 \capr \overline{\prim}_4)\capr \overline{\prim}_3) = ((S_1 \capr \overline{\prim}_3)\capr \overline{\prim}_4)$, but changes the cost, since in general $\mu[S_1 \capr \overline{\prim}_4] \neq \mu[S_2 \capr \overline{\prim}_4]$ and/or $\mu[S_2 \capr \overline{\prim}_3] \neq \mu[S_2 \capr \overline{\prim}_3]$.} \label{fig_plans3D}
\end{figure*}

Figure \ref{fig_plans3D} shows another example in which we alter the design to include a Voronoi lattice structure inside its machined pocket, shown in Fig. \ref{fig_vor}. In addition to the previous 4 primitives (excluding the raw stock), we add another O-AM  primitive to make the lattice, which intersects the original shape as well as the other primitives. The  canonical decomposition need not be recomputed from scratch but can be updated by classifying the pre-existing atoms against the new primitive and splitting them accordingly. This leads to 23 nonempty atoms, in which $A_{11000}$ is classified as fully inside, 14 atoms are classified fully outside, and 8 atoms are classified as partially colliding---only one of them (namely, $A_{10000}$) violates the tolerance specification.
Notwithstanding the non-manufacturable regions, we can construct the closest as-manufactured part to the as-designed target by including the said 1+8 atoms; namely, $A_{11\ast00}$, $A_{111\ast1}$, $A_{011\ast1}$, $A_{01100}$, $A_{11110}$, and $A_{10000}$. The planner is called to map the corresponding minimal DNF or one of its enrichments with $2^5 - 23 = 9$ empty atoms to valid plans. In this case, 6 plans are found to be equivalent with the minimal DNF, while none of the enrichments created new plans.

Assuming relative cost per deposited/removed volume of 1.30, 2.15, 0.85, 0.75, and 1.50 for actions using $\prim_1$ through $\prim_5$, respectively, 6 valid plans were found. The top two plans are the following, also shown in Fig. \ref{fig_plans3D}:
\begin{align*}
	\expr_1[\Prim] &= (( ( ( \prim_1 \capr \overline{\prim}_4 ) \capr \overline{\prim}_3 ) \capr \overline{\prim}_5 ) \cupr \prim_2 ), ~ \text{cost = } 0.3271, \\
	\expr_2[\Prim] &= (( ( ( \prim_1 \capr \overline{\prim}_3 ) \capr \overline{\prim}_4 ) \capr \overline{\prim}_5 ) \cupr \prim_2 ), ~ \text{cost = } 0.3302.
\end{align*} 
The planner only needs atomic encoding to check for the correct outcome and atomic volumes to rank the candidate plans, which are rapid logical/arithmetic operations. Costly geometric operations need not be called at any point after canonical decomposition. Once candidate plans are constructed, they can be validated using physical simulation or domain-specific rules. Once validated, the ROIs can be further analyzed for tool-path planning.

\section {Conclusion}

We have demonstrated an approach to automatic evaluation of manufacturability and generation of process plans for hybrid manufacturing (HM). HM technologies are emerging as the new frontier for fabrication of parts that require complexity and freedom of AM as well as precision and quality of SM, and present a substantial opportunity for computer-aided process planning (CAPP). We have demonstrated that by re-thinking the representation of the as-designed specification in terms of manufacturing primitives, defined as a function of the machine's degrees of freedom (DOF) and minimum manufacturable neighborhoods (MMN), one can evaluate manufacturability prior to the costlier task of process planning by constructing disjoint unions of canonical intersection terms (i.e., atoms) and shortlisting them into an as-manufactured representation. We also showed that the canonical representation can be factorized by purely symbolic manipulations, using a standard AI search algorithm, to a valid as-planned representation. This provides a scalable approach to automatically construct valid, cost-optimal, and qualitatively distinct HM process plans.

This study opens up promising research directions for design for hybrid manufacturing (DfHM). Future directions to investigate include: 1) primitive design targeting localized tolerance specs; 2) advanced planning techniques for Boolean function learning (i.e., DNF-to-plan mapping); and 3) search by iterative fine-graining.

\begin{figure}
	\centering \includegraphics[width=0.45\textwidth]{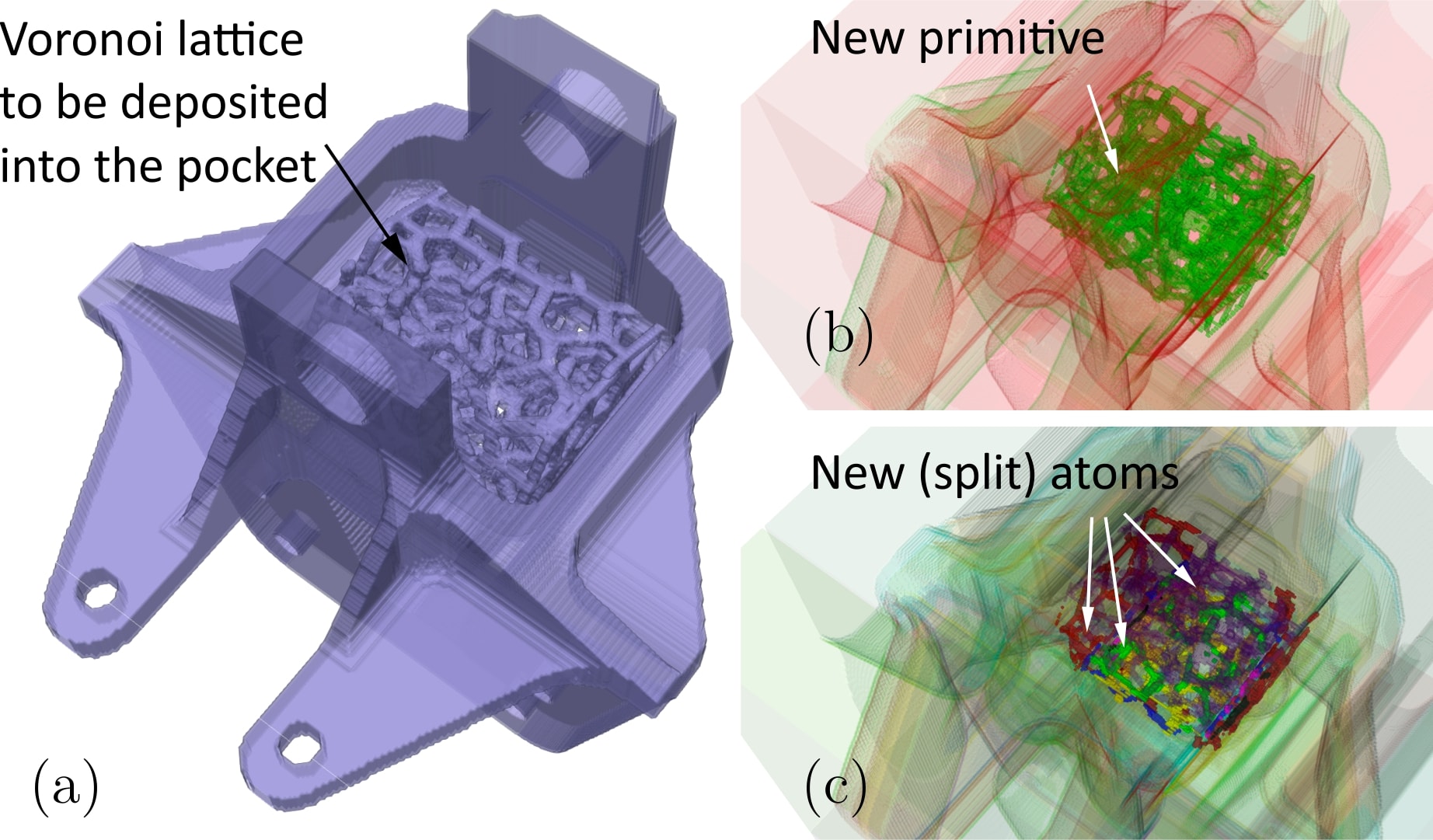}
	\caption{Intersecting an added AM primitive (Voronoi lattice) to the AM/SM primitives of Fig. \ref{fig_prims3D} repeated in (b) splits the 15 nonempty atoms into 23 smaller atoms shown in (c). The new primitive need not be disjoint from the rest of the part in (a).} \label{fig_vor}
\end{figure}

\section*{Acknowledgement}

This research was developed with funding from the Defense Advanced Research Projects Agency (DARPA). The views, opinions and/or findings expressed are those of the authors and should not be interpreted as representing the official views or policies of the Department of Defense or U.S. Government.

\newpage

\bibliography{hybridPlanning}

\end{document}